%% file: main.tex
\documentclass[12pt]{iopart}
\usepackage[T1]{fontenc}
\usepackage{dsfont}
\usepackage{pst-all}
\usepackage{auto-pst-pdf}
\usepackage{comment}
\usepackage{subeqnarray}
\usepackage{amssymb}
\usepackage{epstopdf}
\usepackage{amsthm}
\usepackage{tensor}

\newcommand*{\I}{\mathrm{i}}
\newcommand*{\E}{\mathrm{e}}
\newcommand*{\D}{\mathrm{d}}

\begin{document} 
\title{\boldmath Sums of Random Matrices and the Potts Model on Random Planar Maps}
\author{Max R. Atkin$^{1}$, Benjamin Niedner$^2$ and John F. Wheater$^2$}
\address{$^1$ Universit\'e catholique de Louvain, Chemin du cyclotron 2, B-1348 Louvain-La-Neuve, Belgium}
\address{$^2$ Rudolf Peierls Centre for Theoretical Physics, 1 Keble Road, University of Oxford, Oxford OX1 3NP, UK}
\ead{ \mailto{max.atkin@uclouvain.be}, \mailto{benjamin.niedner@physics.ox.ac.uk}, \mailto{j.wheater@physics.ox.ac.uk}}

\begin{abstract}
We compute the partition function of the $q$-states Potts model on a random planar lattice with $p\leq q$ allowed, equally weighted colours on a connected boundary. To this end, we employ its matrix model representation in the planar limit, generalising a result by Voiculescu for the addition of random matrices to a situation beyond free probability theory. We show that the partition functions with $p$ and $q-p$ colours on the boundary are related algebraically. Finally, we investigate the phase diagram of the model when $0\leq q \leq 4$ and comment on the conformal field theory description of the critical points.
\end{abstract}
\submitto{\jpa}
\maketitle
 
\newtheorem{thm}{Theorem}[section]
\theoremstyle{corollary}
\newtheorem{cor}[thm]{Corollary}
\theoremstyle{lemma}
\newtheorem{lem}[thm]{Lemma}
\theoremstyle{remark}
\newtheorem{rem}[thm]{Remark}
\theoremstyle{example}
\newtheorem{ex}[thm]{Example}
\theoremstyle{definition}
\newtheorem{defn}[thm]{Definition}
\theoremstyle{proposition}
\newtheorem{prop}[thm]{Proposition}

\section{Overview}
\label{sec:overview}

A \em planar map \em can be defined as a graph embeddable in the plane, modulo homeomorphisms. The enumeration of planar maps, initiated by Tutte half a century ago \cite{Tutte63}, has been a topic of long-running interest with important applications in theoretical physics. When endowed with a statistical lattice model defined thereon, the detailed knowledge of the asymptotic properties of large maps affords a rigorous definition of the path integral for two-dimensional gravity coupled to matter or equivalently, bosonic string theory in a non-critical target space dimension \cite{David85,Ambjorn85a,Ambjorn85b,Kazakov85} -- see also \cite{DiFrancesco93,Ginsparg93} for a more general overview. The classification of boundary conditions that can consistently be imposed on the boundary of the graph provides detailed insight into the spectrum of the theory and as a result, the evaluation of the corresponding generating functions has been of continued interest \cite{Carroll96,Carroll97,IR10,BIR10,Chan10,Atkin11,Atkin12}.

Here we compute the partition function $W_{(p)}(z)$ of the $q$-states Potts model \cite{Potts52} on random planar maps with $p\leq q$ allowed, equally weighted colours on a connected boundary, where $z$ denotes the fugacity of a boundary link. For $q=2$ and $q=3$, all of these correspond to boundary conditions of the Ising and Potts lattice models that were found to be integrable on a fixed lattice by Behrend and Pearce \cite{Behrend00}. To this end, we employ the matrix model formulation as proposed long ago by Kazakov \cite{Kazakov88}, defined by a probability measure on $q$ hermitian matrices $\{X_i\}_{i=1}^q$ of size $N\times N$ of the form

\begin{equation}
\label{pottsmeasure}
\D\mu(X_1,X_2,\dots X_q)=\frac{1}{Z_{N,q}}\prod_{\langle i j\rangle}\E^{N\mathrm{tr}X_i X_j}\times \prod_{i=1}^q\E^{-N\mathrm{tr}V_i(X_i)}\D X_i\ ,
\end{equation}

\noindent where $\D X_i$ denotes the integration over independent components,

\begin{equation}
\mathrm{d}X_i=\prod_{1\leq a\leq b \leq N}\D \mathrm{Re}(X_i)\indices{^a_b}\prod_{1\leq a<b \leq N}\D\mathrm{Im}(X_i)\indices{^a_b}\ .
\end{equation}

\noindent Here, the partition function $Z_{N,q}$ normalises expectation values such that $\int\D\mu=1$ and $\langle ij \rangle$ denotes the product over distinct $i,j$. In the particular case of random \em triangulations\em, $W_{(1)}(z)$ was first found by Daul \cite{Daul94} and later Zinn-Justin \cite{ZJ99} in the saddle point approximation for integer $0\leq q\leq 4$, and by Bonnet and Eynard \cite{Bonnet99,Eynard99} using the method of loop equations; the latter also found an algebraic equation for $W_{(1)}(z)$ when $\arccos((q-2)/2)/\pi$ is rational. See also \cite{Guionnet10,Bernardi11} for related results. More recently, the authors of \cite{Borot12} considered a combinatorial approach using the so-called ``loop-gas'' representation of the Potts model on planar maps without reference to a matrix integral, from which a pair of coupled functional equations for $W_{(1)}(z)$ and a function related to $W_{(q)}(z)$ was obtained and solved. For $q=2$, the relationship between $W_{(1)}(z)$ and $W_{(2)}(z)$ has been expressed succinctly from the perspective of the boundary renormalisation group \cite{Carroll96,Carroll97}, a picture which later was extended to non-planar geometries \cite{Atkin11} and arbitrary face degrees \cite{Atkin12}. Indeed, these investigations revealed that different boundary conditions yield inequivalent algebraic equations satisfied by the corresponding generating functions. However, a systematic understanding of the relationship between different boundary conditions for more general values of $q$ and $p$ appears to be lacking and herein we report on some progress on this matter. 
 
 As will be discussed in Section \ref{sec:bc}, $W_{(p)}(z)$ is given by the Stieltjes transform of the spectral density of the sum of $p$ hermitian random matrices of infinite size. A more mathematically inclined characterisation of the problem solved in this paper thus goes as follows: given a set of $q$ hermitian random matrices distributed according to \eref{pottsmeasure} and a positive integer $p\leq q$, what is the spectral density of the sum $X_1+X_2+\dots X_p$ as $N\to\infty$? For the simpler case of uncorrelated matrices, the answer has been succinctly summarised in the context of free probability  \cite{Speicher93,Speicher94}, going back to Voiculescu's observation of the asymptotic freeness of Gaussian independent random matrices \cite{Voiculescu92}: given a spectral density $\rho_X(z)$, define the \em R-transform \em via the functional inverse of its Stieltjes transform $W_X(z)$,

\begin{equation}
\label{rtrafo} 
R^X(z)=W_X^{-1}(z)-\frac{1}{z}\ .
\end{equation}

\noindent Now assume $Y$ is freely independent from $X$. Then the \em free (additive) convolution \em $\rho_X\boxplus\rho_Y$ is defined by $\rho_{X+Y}$. The results of free probability theory \cite{Voiculescu92} state that the latter is obtained from $\rho_X$ and $\rho_Y$ by

\begin{enumerate}
\item computing $R^{X+Y}$ by adding the respective R-transforms,
\begin{equation}
R^{X+Y}(z)=R^X(z)+R^Y(z)\ ,
\end{equation}
\item inverting the relationship \Eref{rtrafo}, 
\begin{eqnarray}
\label{Voicu}
W_{X+Y}^{-1}(z)&=R^{X+Y}(z)+\frac{1}{z}\ .
\end{eqnarray}
\end{enumerate}

\noindent The spectral density for the sum $X+Y$ can then be read off from the imaginary part of the inverse function,
\begin{eqnarray}
\rho_{X+Y}(x)&=\frac{1}{\pi}\mathrm{Im}\ W_{X+Y}(x)_+\ .
\end{eqnarray}

\noindent We follow \cite{ZJ99-2} in referring to the key relationship given by \eref{Voicu} as Voiculescu's formula. Clearly the matrices $\{X_i\}_{i=1}^q$ distributed according to \eref{pottsmeasure} are \em not \em freely independent -- their correlations prevent us from applying Voiculescu's formula to compute the spectral densities for sums like $X_1+X_2+\dots X_p$. Our strategy to obtain the disk partition function of the Potts model involves a suitable generalisation of the R-transform and using it to evaluate the spectral density and hence $W_{(p)}(z)$. 

This article is organised as follows: Section \ref{sec:bc} reviews the matrix model formulation and defines the observables of interest. In Section \ref{sec:planar}, we state the main results in Propositions \ref{prop-curve} and \ref{prop-ycurve} and discuss how our results reduce to Voiculescu's formula when the interactions of the Potts model are turned off. In Section \ref{sec:cases}, we study hard dimers, the Ising model and the $3$-states Potts model on planar triangulations as simple examples in greater detail, deriving explicit expressions for the spectral curve for given $p$ and comparing our results to the literature where available. In Section \ref{sec:crit}, we proceed to investigate the phase diagram of the model when $0<q<4$ and comment on the conformal field theory description of the scaling behaviour associated with the critical points. Finally, we discuss the implications of our results in Section \ref{sec:pottsdiscussion}.

\section{Definition of the model}
\label{sec:bc}

\noindent Following \cite{Kazakov88,Daul94,ZJ99,Bonnet99,Eynard99}, we use \eref{pottsmeasure} to define observables of the $q$-states Potts model on a random planar lattice. A distinguishing feature of this measure are the exponentials of $\mathrm{tr}X_i X_j$ in \eref{pottsmeasure}, breaking the overall $U(N)\times O(q)$-invariance of the remaining factors. Here we confine our study to the case $V_i(z)=U(z)+z^2/2$ $\forall i$ for a fixed polynomial $U(z)=\sum_{m=2}^{k+2}t_mz^m/m$, rendering the $q$ states of the statistical system indistinguishable. In this case, the measure remains invariant under the overall symmetries

\begin{equation}
\label{sym}
	X_i\to U^{\dagger}X_iU\ ,\quad U\in U(N)\ ,\quad\mathrm{and}\quad X_i\to X_{\sigma(i)}\ ,\quad \sigma\in S_q\ ,
\end{equation}

\noindent where $S_q$ denotes the symmetric group of order $q!$. This is to be contrasted with the ``multi-matrix chain'' studied for example in \cite{ZJ98,Eynard03}, for which $\mathbb{Z}_2$ is preserved in place of $S_q$. Our definition includes a subset of the statistical RSOS models on a random lattice, which are indexed by simply laced Dynkin diagrams \cite{Pasquier86} and have been described using matrix integrals by Kostov \cite{Kostov95}. In particular, for $(q,k)=(2,1)$, \eref{pottsmeasure} describes the $A_3$ model and for $(q,k)=(3,1)$ the $D_4$ model on random triangulations, respectively.

  The desired quantities $W_{(p)}(z)$ can now be defined as follows: Given $\sigma\in S_q/(S_p\times S_{q-p})$, we define the partition function of the model on a random lattice with $p$ allowed, equally weighted colours on a single connected boundary containing a marked point as

\begin{eqnarray}
\label{discfct}
	W_{(p|\sigma)}(z)=\frac{1}{N}\left\langle\mathrm{tr}\frac{1}{z-X_{(p|\sigma)}}\right\rangle\ ,\quad X_{(p|\sigma)}=\sum_{i=1}^pX_{\sigma(i)}\ ,
\end{eqnarray}

\noindent where here and in what follows $\langle\cdot\rangle$ denotes the average with respect to the measure \eref{pottsmeasure} and $z$ denotes the fugacity of a boundary link. As a result of the permutation symmetry, for given $p$, all $|S_q/(S_p\times S_{q-p})|={q \choose p}$ partition functions $W_{(p|\sigma)}(z)$ are described by the same function, so that we henceforth abbreviate $W_{(p)}(z):=W_{(p|\sigma_0)}(z)$ for a representative $\sigma_0$ and denote the spectral density of the sum $X_{(p|\sigma)}$ by $\rho_{(p)}(z)$. Note that for $p=1$, our definition of $W_{(p)}(z)$ reduces to the one studied in \cite{Kazakov88,Daul94,ZJ99,Bonnet99,Eynard99}. 

We conclude this section with a helpful lemma which expresses the partition function as single integral over effective matrix variables $X_0$ and $P_\pm$ by a series of integral transformations. Thinking of $\{X_i\}_{i=1}^q$ as coordinates on configuration space, these simply correspond to linear canonical transformations on the corresponding phase space. This circumvents a notorious difficulty presented by the measure \eref{pottsmeasure}, which when written as a function of the eigenvalues of the matrices $X_i$, $i>0$, leads to a complicated integral over the unitary group \cite{Kazakov88,Daul94}. 

\begin{lem} 
\label{lem-Z} Let $h>0$ and abbreviate the integral transformations

\begin{eqnarray}
\gamma_{\pm}(X)&=\int_{\mathbb{R}}\D P_\pm f(P)\E^{-\frac{N}{2}\mathrm{tr}P_\pm^2}\E^{N\mathrm{tr}P_\pm X/\sqrt{\E^{\pm 2h}-1}}\ ,\\
\gamma'_{\pm}(P)&=\int_\Gamma\D X f(X)\E^{N\mathrm{tr}PX\sqrt{1-\E^{\mp 2h}}}\ ,
\end{eqnarray}

\noindent where the subscripts below the integrals indicate the integration cycle for the corresponding eigenvalues. Then up to an overall constant, the partition function in \eref{pottsmeasure} can be written as

\begin{eqnarray}
	\label{Z1}
Z_{N,q}
	&=\int_{\mathbb{R}}\mathrm{d}P_+\ \E^{-\frac{N}{2}(1-\E^{-2h})\mathrm{tr}P_+^2}\left(\gamma'_+\left[\E^{-N\mathrm{tr}U}\right](P_+)\right)^q\\
	\label{Z2}
	&=\int_\mathbb{R}\mathrm{d}X_0\ \gamma_+\left[(\gamma'_+[\E^{-N\mathrm{tr}U}])^p\right](X_0)\ \gamma_-\left[(\gamma'_-[\E^{-N\mathrm{tr}U}])^{q-p}\right](X_0)\ \\
	\label{Z3}
	&=\int_\mathbb{R}\mathrm{d}X_0\left(\prod_{i=1}^q\int_\Gamma\mathrm{d}X_i\ \E^{-N\mathrm{tr}U(X_i)}\right)\nonumber\\&\quad\times\gamma_+[1]\left(X_0+2\sinh(h)\sum_{i=1}^pX_i\right)\nonumber\\&\quad\times\gamma_-[1]\left(X_0-2\sinh(h)\sum_{i=p+1}^q X_i\right)\ .
\end{eqnarray}
\end{lem}

\begin{proof}
We begin by showing equality of \eref{Z2} and \eref{Z3}, and then equality to $Z_{N,q}$. Subsequently showing equality to \eref{Z1} completes the proof. By definition, we can write

\begin{eqnarray}
\label{integralexp}
\gamma_{\pm}\left[\left(\gamma'_\pm[\E^{-N\mathrm{tr}U}]\right)^n\right](X_0)&=\int_\mathbb{R}\mathrm{d}P_\pm\ \E^{-N\mathrm{tr}P^2/2}\ \E^{N\mathrm{tr}X_0 P_\pm/\sqrt{\E^{\pm 2h}-1}}\nonumber\\&\quad\times\left(\prod_{i=1}^n\int_\Gamma\mathrm{d}X_i\ \E^{-N\mathrm{tr}U(X_i)}\ \E^{N\mathrm{tr}P_\pm X_i\sqrt{1-\E^{\mp 2h}}}\right)\ .
\end{eqnarray}

\noindent In general, there are $\mathrm{deg}\ U'=k+1$ independent cycles $\Gamma$ that render this iterated integral absolutely convergent for finite $N$. Hence we can apply the Fubini-Tonelli-theorem, that is, exchange the order of integration:

\begin{eqnarray}
\gamma_{\pm}\left[\left(\gamma'_\pm[\E^{-N\mathrm{tr}U}]\right)^n\right](X_0)
	&= \left(\prod_{i=1}^n\int_\Gamma\mathrm{d}X_i\ \E^{-N\mathrm{tr}U(X_i)}\right)\nonumber\\&\quad\times\gamma_\pm[1]\left(X_0\pm2\sinh(h)\sum_{i=1}^n X_i\right)\ .
\end{eqnarray}

\noindent Inserting this result into \eref{Z2} proves equality to \eref{Z3}. To obtain equality to $Z_{N,q}$, note that up to an overall multiplicative constant,

\begin{eqnarray}
\E^{N\mathrm{tr}\left(\sum_{i=1}^qX_i\right)^2/2}&=\int_\mathbb{R}\mathrm{d}X_0\ \gamma_+[1]\left(X_0+2\sinh(h)\sum_{i=1}^p X_i\right)\nonumber\\&\quad\times\gamma_-[1]\left(X_0-2\sinh(h)\sum_{i=p+1}^q X_i\right)\ .
\end{eqnarray}

\noindent Inserting this result into \eref{Z3} and interchanging the order of integration between $X_0$ and $X_i$ by the same argument proves equality to $Z_{N,q}$. It remains to show equivalence to \eref{Z1}. Starting from \eref{Z3}, we may use \eref{integralexp} to write the action of $\gamma_+$ on $(\gamma'_+[\E^{-N\mathrm{tr}U}])^p$ and of $\gamma_-$ on $(\gamma'_-[\E^{-N\mathrm{tr}U}])^{q-p}$ as Gaussian integrals over two matrices $P_+$, $P_-$, respectively. Performing the integration over $P_-$ and subsequently $X_0$, we are left with \eref{Z1}. As a cross-check, it is straightforward to confirm that \eref{Z1} equals our initial definition of $Z_{N,q}$
by writing out the $q^{\mathrm{th}}$ power of $\gamma'_+[\E^{-N\mathrm{tr}U}]$ as a product of integrals over $X_i$, $i=1\dots q$ and reversing the order of integration with $P_+$.
\end{proof}

\section{Planar limit}
\label{sec:planar}

\noindent This section is concerned with the explicit evaluation of $W_{(p)}(z)$ in the planar limit and is organised as follows: Subsection \ref{subsec:saddle} expresses $W_{(p)}(z)$ via the $p$-independent spectrum of the matrix $Y\equiv \sqrt{1-\E^{-2h}}P_+$  appearing in \eref{Z1}. The spectrum is then obtained explicitly in Subsection \ref{subsec:solution} -- a problem first solved in \cite{Daul94,ZJ98} and rederived here for arbitrary $q\neq 4$. Finally, in Subsection \ref{subsec:voiculescu}, we discuss how Voiculescu's formula \eref{Voicu} arises as a special case of our results in the limit of vanishing interaction strength of the Potts model. For any hermitian matrix $X$ with eigenvalues $\{x_i\}_{i=1}^N$, we define the large-$N$ spectral density

\begin{equation}
\rho_X(x)=\lim_{N\to\infty}\frac{1}{N}\left\langle\sum_{i=1}^N\delta(x-x_i)\right\rangle\ ,
\end{equation}

\noindent and introduce its Stieltjes transform

\begin{equation}
W_X(z)=\int_{\mathrm{supp}\ \rho_X}\mathrm{d}x\frac{\rho_X(x)}{z-x}\ ,\quad z\notin\mathrm{supp}\ \rho_X\ .
\end{equation}

\noindent To streamline the presentation, for any pair of $N\times N$ matrices $(X,P)$, we also define the averages

\begin{eqnarray}
\label{def-g}
	G^X_P(z)&=\frac{1}{N}\frac{\partial}{\partial z}\ln\left\langle\det_{1\leq k, l\leq N}\E^{Nx_kp_l}\right\rangle_{p_N=z}\ ,\quad z\notin \mathrm{supp}\ \rho_P\ ,\\
	G^P_X(z)&=\frac{1}{N}\frac{\partial}{\partial z}\ln\left\langle\det_{1\leq k, l\leq N}\E^{Nx_kp_l}\right\rangle_{x_N=z}\ ,\quad z\notin \mathrm{supp}\ \rho_X\ .
\end{eqnarray} 

\noindent The key property of the above functions is that $G^X_P(G_X^P(z))=z\left(1+\mathcal{O}(1/N)\right)$ for large $N$ \cite{Matytsin93,ZJ98}.

\subsection{Saddle point equations}
\label{subsec:saddle}

\noindent We begin by stating the main result of this section, Proposition \ref{prop-curve}. This rests on Lemmas \ref{lem-loop} and \ref{lem-SPE}, which we then derive 
before concluding with the proof of the proposition.

\begin{prop} 
\label{prop-curve}
Let the random matrix $P_+$ be defined as in Lemma \ref{lem-Z}, and set $Y=\sqrt{1-\E^{-2h}}P_+$. Then for $N\to\infty$, the spectral density of the sum of $p$ matrices distributed according to \eref{pottsmeasure} is given by\footnote{Throughout we denote by $f(z)_+$ the value of a multivalued function $f(z)$ on the physical sheet, and by $f(z)_-$ the value on the sheet connected to the physical sheet by the physical cut.}

\begin{equation}
\label{curve1}
\rho_{(p)}(z)=\frac{1}{2\pi\I}\left[G^Y_{(p)}(z)_+-G^Y_{(p)}(z)_-\right]\ ,
\end{equation} 
\noindent where $G^Y_{(p)}(z)$ is the functional inverse of 
\begin{equation}
\label{curve2}
G_Y^{(p)}(z)=\frac{p}{q}(z-W_Y(z)_-)+\frac{q-p}{q}W_Y(z)_+\ .
\end{equation}
\end{prop}

\begin{rem} Generally, $G^Y_{(p)}(z)$ is a multi-valued function so that we need to specify the sheet on which \eref{curve1} is evaluated -- this ambiguity is fixed by requiring $\lim_{z\to\infty}z\ W_{(p)}(z)=1$. 
\end{rem}

\begin{cor} 
\label{curve3}
When $G_{(p)}^Y(z)$ satisfies an algebraic equation of the form $F_{(p)}(z,G_{(p)}^Y(z))=0$, then $G_{(q-p)}^Y(z)$ follows from

\begin{equation}
F_{(p)}\left(G_{(q-p)}^Y(z)-z,G_{(q-p)}^Y(z)\right)=0\ .
\end{equation}
\end{cor}

\noindent As announced, we proceed to formulate the main Lemmas involved in the proof of the above results:

\begin{lem} 
\label{lem-loop} In the limit $N\to\infty$, the matrix $M=\E^{-h}\sum_{i=1}^pX_i+\E^{h}\sum_{i=p+1}^qX_i$ satisfies
\begin{equation}
W_{X_0}(z)=W_M(z-W_{X_0}(z))\ .
\end{equation}
\end{lem}

\begin{proof}
This result follows from the translation invariance of the measure $\mu$ defined in \eref{pottsmeasure}. Setting $\sigma=\mathrm{id}$ in \eref{Z3} without loss of generality, consider the shift by a small hermitian matrix\footnote{See also \cite{Eynard99} for an earlier application of this method of ``loop equations'' to the Potts matrix model.}

\begin{equation}
X_0\longrightarrow X_0'=X_0+\varepsilon\left(\frac{1}{z-X_0}\frac{1}{z'-M}+\mathrm{h.c.}\right)\ ,\quad \varepsilon\ll 1\ ,
\end{equation}

\noindent understood as a formal power series in $z$, $z'$. When $M=\E^{-h}\sum_{i=1}^pX_i+\E^{h}\sum_{i=p+1}^qX_i$, the variation of the product of the two gaussian integrals 

\begin{eqnarray}
I\left(\{X_i\}_{i=0}^q\right)&\equiv\gamma_+[1]\left(X_0+2\sinh(h)\sum_{i=1}^pX_{i}\right)\nonumber\\&\quad\times \gamma_-[1]\left(X_0-2\sinh(h)\sum_{i=p+1}^q X_{i}\right)
\end{eqnarray}

\noindent and the measure $\mathrm{d}X_0$ is respectively given to leading order by

\begin{eqnarray}
I\left(\{X_i\}_{i=0}^q\right)& \longrightarrow I(\{X_i\}_{i=0}^q)+\varepsilon\mathrm{tr}\frac{1}{z-X_0}\frac{1}{z'-M}(M-X_0)+\mathcal{O}(\varepsilon^2)\ ,\\
\mathrm{d}X_0 & \longrightarrow \mathrm{d}X_0\left(1+\varepsilon\mathrm{tr}\frac{1}{z-X_0}\mathrm{tr}\frac{1}{z-X_0}\frac{1}{z'-M}+\mathcal{O}(\varepsilon^2)\right)\ .
\end{eqnarray}

\noindent Demanding invariance of $Z_{N,q}$ to leading order in $\varepsilon$ and approximating $\langle\mathrm{tr}A\mathrm{tr}B\rangle=\langle\mathrm{tr}A\rangle\langle\mathrm{tr}B\rangle+\mathcal{O}(1/N)$ yields, as $N\to\infty$,

\begin{equation}
W_M(z')-W_{X_0}(z)=\frac{1}{N}\left\langle\mathrm{tr}\frac{1}{z-X_0}\frac{1}{z'-M}\right\rangle\left(W_{X_0}(z)-z+z'\right)\ .
\end{equation}

\noindent Evaluating the above at $z'=z-W_{X_0}(z)$ proves the Lemma.
\end{proof}

\begin{lem} 
\label{lem-SPE}
Let the matrices $P_+$, $X_0$ be defined as in Lemma \ref{lem-Z}, and set $P_+=Y/\sqrt{1-\E^{-2h}}$ and $X_0=2\sinh(h)\bar{X}_0$. Then for $z\in\mathrm{supp}\ \rho_Y$ as $N\to\infty$,

\begin{equation}
\label{SPE}
\mathrm{Re}\ G_Y^{\bar{X}_0}(z)=\left(\frac{2p}{q}-1\right)\mathrm{Re}\ W_Y(z)+\left(\frac{1}{1-\E^{-2h}}-\frac{p}{q}\right)z\ .
\end{equation}
\end{lem}

\begin{proof} 
This result follows from the saddle point approximation applied to \eref{Z1} and \eref{Z2} in Lemma \ref{lem-Z}. Setting 

\begin{equation}
\label{diagYX0}
Y=U\mathrm{diag}(\{y_n\}_{n=1}^N)U^{\dagger}\ , \quad \frac{X_0}{2\sinh(h)}=V\mathrm{diag}(\{x_n\}_{n=1}^N)V^{\dagger}\ ,
\end{equation}

\noindent with $U, V \in U(N)$, we can perform the integration over $U^{\dagger}V$ using the well-known result \cite{Harishchandra56,Itzykson80} 

\begin{equation}
\label{IZint}
	\int_{U(N)}\mathrm{d}U\ \E^{\lambda N\mathrm{tr}[YUXU^{\dagger}]}=\mathrm{const.}\times \frac{\det_{1\leq k,l\leq N}\E^{\lambda Ny_kx_l}}{\Delta(x)\Delta(y)}\quad\forall\lambda\in\mathbb{C}\ ,
\end{equation}

\noindent where $\mathrm{d}U$ is the normalised Haar measure. It follows that, for the exponent of the integrand in \eref{Z2} to have an extremum, the eigenvalues $\{x_n\}$ and $\{y_n\}$ 
must satisfy

\begin{eqnarray}
\label{SPE1}
0&=\frac{1}{N}\left(\frac{\partial}{\partial y_n}\ln\det_{k,l}\E^{Ny_kx_l}+\sum_{k\neq n}\frac{1}{y_n-y_k}\right)-\frac{y_n}{1-\E^{-2h}}\nonumber\\&\quad+\frac{p}{N}\frac{\partial}{\partial y_n}\ln\gamma'_+[\E^{-N\mathrm{tr}U}]\left(\frac{Y}{\sqrt{1-\E^{-2h}}}\right)\ .
\end{eqnarray}

\noindent On the other hand, from \eref{Z1} and \eref{SPE1} we find

\begin{equation}
\label{SPE3}
0=\frac{2}{N}\sum_{k\neq n}\frac{1}{y_n-y_k}+\frac{q}{N}\frac{\partial}{\partial y_n}\ln\gamma'_+[\E^{-N\mathrm{tr}U}]\left(\frac{Y}{\sqrt{1-\E^{-2h}}}\right)-y_n\ ,
\end{equation}

\noindent which allows us to eliminate $\gamma'_+[\E^{-N\mathrm{tr}U}](Y/\sqrt{1-\E^{-2h}})$ between the above and \eref{SPE1}\footnote{It is the analysis of this quantity that leads us to the exact solution for $W_Y(z)$ in the next subsection.}. Taking $N\to\infty$ and using \eref{def-g} yields \eref{SPE} as advertised. 
\end{proof}

\begin{proof}[Proof of Proposition \ref{prop-curve} and Corollary \ref{curve3}] 
\Eref{curve2} follows from Lemma \ref{lem-SPE} after analytic continuation. Following an argument in \cite{ZJ99}, we note that the derivative w.r.t. $x_N$ of the logarithm of \eref{IZint} is an entire function of $x_N$, which implies that as $N\to\infty$, $G_X^Y(z)$ and $W_X(z)$ have the same discontinuity across the real axis. Applying this to our situation, we conclude that when $z\in\mathrm{supp}\ \rho_Y$,

\begin{eqnarray}
\label{discont1}
G_Y^{\bar{X}_0}(z)_{\pm}&=\mathrm{Re}\ G_Y^{\bar{X}_0}(z)\pm \I\pi\rho_Y(z)\ ,\\
\label{discont2}
G^Y_{\bar{X}_0}(z)_{\pm}&=\mathrm{Re}\ G^Y_{\bar{X}_0}(z)\pm \I\pi\rho_{\bar{X}_0}(z)\ .
\end{eqnarray}

\noindent For $h>0$, it follows from \eref{Z1} that $Y$ is hermitian, so $W_Y(z)$ has no singularities in the complex plane away from the real axis. Hence we can analytically continue \eref{SPE} to $z\in\mathbb{C}\setminus\mathrm{supp}\ \rho_Y$ using

\begin{equation}
G_Y^{\bar{X}_0}(z)_+-G_Y^{\bar{X}_0}(z)_-=W_Y(z)_+-W_Y(z)_-\ ,
\end{equation}

\noindent which results in

\begin{equation}
\label{spe-complex}
G_Y^{\bar{X}_0}(z)=\frac{p}{q}\left(W_Y(z)_+-z\right)-\frac{q-p}{q}W_Y(z)_-+\frac{z}{1-\E^{-2h}}\ .
\end{equation}

\noindent To obtain \eref{curve1}, note first that from Lemma \ref{lem-loop}, we find

\begin{eqnarray}
W_{\bar{X}_0}(z)=W_{M/(2\sinh(h))}\left(z-\frac{1}{4\sinh(h)^2}W_{\bar{X}_0}(z)\right)\ ,
\end{eqnarray}

\noindent where we used the property $W_X(z)=\lambda W_{\lambda X}(\lambda x)$ for real $\lambda$. In the limit $h\to\infty$, $M/(2\sinh(h))\to\sum_{i=p+1}^qX_i$ and consequently, from the above,

\begin{equation}
W_{\bar{X}_0}(z)\to W_{(q-p)}\left(z+\mathcal{O}(\E^{-2h})\right)\quad \mathrm{as}\ h\to\infty\ .
\end{equation}

\noindent We infer that in this limit, $\rho_{\bar{X}_0}(z)\to \rho_{(q-p)}(z)$, which in conjunction with \eref{discont2} yields

\begin{equation}
\rho_{(q-p)}(z)=\lim_{h\to\infty}\frac{1}{2\pi\I}\left[G^Y_{\bar{X}_0}(z)_+-G^Y_{\bar{X}_0}(z)_-\right]\ .
\end{equation}

\noindent We thus obtain the desired expressions \eref{curve1} and \eref{curve2} from the above and \eref{spe-complex} by identifying  $G^Y_{\bar{X}_0}(z)=G_{(q-p)}^Y(z)$, and noting that according to \eref{spe-complex}, the analytic continuation of $G_Y^{(q-p)}(z)$ through $\mathrm{supp}\ \rho_Y$ is given by $z-G_Y^{(p)}(z)$,

\begin{equation}
\label{continuationGp}
G_Y^{(q-p)}(z)_\pm=z-G_Y^{(p)}(z)_\mp\ .
\end{equation}
 
\noindent The functional inversion relation then follows from $G^Y_{\bar{X}_0}\circ G_Y^{\bar{X}_0}=\mathrm{id}$. Finally, to show Corollary \ref{curve3}, observe that according to \eref{continuationGp}, for an algebraic function $F_{(p)}$ in two variables,

\begin{equation}
F_{(p)}\left(z,G_{(p)}^Y(z)\right)=0\quad \mathrm{implies}\quad F_{(p)}\left(z'-G_Y^{(q-p)}(z'),z'\right)=0\ ,
\end{equation}

\noindent since the analytic continuation merely takes us from one solution to the above equation to another. Evaluating at $z'=G^Y_{(q-p)}(z)$ proves the Corollary.
\end{proof}

\subsection{General solution}
\label{subsec:solution}

\noindent The main result in Proposition \ref{prop-curve} is expressed via the functional inverse of $G_Y^{(p)}(z)$.
Generally, this functional inversion is most easily achieved by means of an explicit parametric form of $W_Y(z)$; Proposition \ref{prop-ycurve} below provides just that for general $q\neq 4$ when $U(z)$ is cubic, i.e. $k=1$. Our conventions for elliptic functions are those of Gradshtein and Ryzhik \cite{Gradshteyn07} and are spelled out in \ref{app:auxprob}.

\begin{prop} 
\label{prop-ycurve}
Let $k=1$, $\nu=\arccos((q-2)/2)/\pi$ and assume $\mathrm{supp}\ \rho_Y=[z_-,z_+]\subset\mathbb{R}$ as $N\to\infty$. Then

\begin{equation}
W_Y(z(\sigma))=\frac{1}{4-q}\left(\frac{qt_2}{t_3}+2z(\sigma)\right)+\sum_{n\geq 0}\frac{f_n}{n!}\frac{
		\partial^n}{\partial\sigma^n}\left[g(\sigma;\nu)+g(\sigma;-\nu)\right]\ ,
\end{equation}

\noindent with the coefficients $\{f_n\}$ determined by the requirement $\lim_{z\to\infty}zW_Y(z)=1$ and
	
\begin{eqnarray}
	z(\sigma)&=\delta_U+\sqrt{(z_+-\delta_U)(z_--\delta_U)}\left(\frac{\vartheta_2(\pi\sigma|\tau)}{\vartheta_3(\pi\sigma|\tau)}\right)^2\ ,\\
\label{g-function}
		g(\sigma;\nu)&=\E^{\I\pi\nu\sigma}\frac{\vartheta_3(\pi\sigma+\pi\tau\nu/2|\tau)}{\vartheta_3(\pi\sigma|\tau)}\ ,
	\end{eqnarray}

\noindent where $\tau$ and $\delta_U$ are implicit functions of $t_2$, $t_3$.
\end{prop}

\begin{proof} 
We can determine the spectrum of $Y$ from the saddle point equation \Eref{SPE3}
which is precisely the problem first considered in \cite{Daul94,ZJ99}. To our knowledge, the first large $N$ analysis of $\gamma'_+[\E^{-N\mathrm{tr}U}]$ 
for cubic $U$ was done, if in a slightly different context, by Gross and Newman in \cite{Gross91}. Using \cite[eqns. (2.10), (2.11)]{Gross91}, equation \eref{SPE3} can be expressed as 

\begin{eqnarray}
	\label{spe1}
		z=2\mathrm{Re}\ W_Y(z)&+\frac{q}{2}\int_{z_-}^{z^+}\frac{\mathrm{d}z'}{\sqrt{z'-\delta_U}}\frac{\rho_Y(z')}{\sqrt{z-\delta_U}+\sqrt{z'-\delta_U}}\nonumber\\ &+q\frac{\sqrt{z-\delta_U}}{\sqrt{t_3}}-q\frac{t_2}{2t_3}\ ,\qquad z\in[z_-,z_+]\ ,
	\end{eqnarray}

\noindent where $\delta_U$ solves the implicit equation 

\begin{equation}
\frac{t_2}{4t_3}+\delta_U=\frac{\sqrt{t_3}}{t_2^{3/2}}\int_{z_-}^{z^+}\mathrm{d}z\ \frac{\rho_Y(z)}{\sqrt{z-\delta_U}}\ .
\end{equation}

\noindent Let us resolve the branch point at $\delta_U$ by the change of variables $w(z)=\sqrt{z-\delta_U}$, and denote $w(z_{\pm})=w_{\pm}$. Introducing the auxiliary function

\begin{equation}
	\label{auxfct}
	f(w)=\int_{w_-}^{w_+}\mathrm{d}w'\frac{\rho_Y(\delta_U+w^{\prime 2})}{w-w'}\ ,
\end{equation}

\noindent we derive the two identities

\begin{eqnarray}
		\mathrm{Re}\ W_Y(z)&=\mathrm{Re}\ f(w(z))+f(-w(z))\;,\qquad z\in[z_-,z_+]\;,\\
		f(-w(z))&=-\frac{1}{2}\int_{z_-}^{z_+}\frac{\mathrm{d}z'}{\sqrt{z'-\delta_U}}\frac{\rho_Y(z')}{\sqrt{z-\delta_U}+\sqrt{z'-\delta_U}}\;.
	\end{eqnarray}

\noindent We can then rewrite \eref{spe1} in the equivalent form
	
\begin{equation}
	2\mathrm{Re}\ f(w)+(2-q)f(-w)=\delta_U+w^2-q\frac{w}{\sqrt{t_3}}+q\frac{t_2}{2t_3}\ , \quad w\in[w_-,w_+]\ .
\end{equation}

\noindent A particular polynomial solution to the above equation when $q\neq 4$ is

\begin{equation}
	f_{\mathrm{r.}}(w)=\frac{qt_2}{2t_3(4-q)}-\frac{w}{\sqrt{t_3}}+\frac{\delta_U+w^2}{4-q}\ .
\end{equation}

\noindent The general solution will therefore differ from the above by a function $f_{\mathrm{s.}}(w)=f_{\mathrm{r.}}(w)-f(w)$ holomorphic on $\mathbb{C}\setminus[w_-,w_+]$ satsifying the homogenous equation

\begin{equation}
	\label{hom1}
	2\mathrm{Re}\ f_{\mathrm{s.}}(w)+(2-q)f_{\mathrm{s.}}(-w)=0\ , \quad w\in[w_-,w_+]\ .
\end{equation}

\noindent We recover $\rho_Y(z)$ by inverting \eref{auxfct}, which, using the fact that $f_{\mathrm{r.}}(w)$ is analytic, gives
	
\begin{eqnarray}
	\label{lambdadens}
		\rho_Y(z)&=\frac{1}{2\pi\I}\left[f_{\mathrm{s.}}\left(\sqrt{z-\delta_U}\right)_--f_{\mathrm{s.}}\left(\sqrt{z-\delta_U}\right)_+\right]\ .
	\end{eqnarray}
	
\noindent From the above expressions it then follows that for $z\notin[z_-,z_+]$, $W_Y(z)$ is given by

\begin{eqnarray}	
	\label{lambdares}	W_Y(z)
		&=\frac{1}{4-q}\left(\frac{qt_2}{t_3}+2z\right)-f_{\mathrm{s.}}\left(\sqrt{z-\delta_U}\right)-f_{\mathrm{s.}}\left(-\sqrt{z-\delta_U}\right)\ .
	\end{eqnarray}

\noindent The general solution to \eref{hom1} was first derived in \cite{Eynard95} in the context of the $O(n)$ model on a random lattice and is presented in more detail in \ref{app:auxprob}. There we recall how $f_{\mathrm{s.}}(w)$ can be parametrised in terms of $g(\sigma;\nu)$, defined in \eref{g-function}, as\footnote{By abuse of notation, we distinguish the functions $w(\sigma)$ and $w(z)=\sqrt{z-\delta_U}$ solely by their arguments.}
	
\begin{eqnarray}
	\label{ellipticsol}
		f_{\mathrm{s.}}(w(\sigma))&=\sum_{n\geq 0}\frac{f_n}{n!}\frac{
		\partial^n}{\partial\sigma^n}\left[\E^{-\I \pi \nu}g(\sigma;\nu)+\E^{\I\pi\nu}g(\sigma;-\nu)\right]\ ,\\
w(\sigma)&=\sqrt{w_-w_+}\frac{\vartheta_2(\pi\sigma|\tau)}{\vartheta_3(\pi\sigma|\tau)}\ ,
	\end{eqnarray}
	
\noindent where $\nu=\arccos((2-q)/2)/\pi$ and $\tau=\I K'/K$, with $K$ and $K'$ respectively given by the complete elliptic integral of the first and second kind (cf. \eref{ellipticK}); the coefficients $\{f_n\}$ are entirely determined by the condition that $\lim_{z\to\infty} z\ W_Y(z)=1$. Inserting the above parametrisation into \eref{lambdares} completes the proof.
\end{proof}

\subsection{Derivation of Vouiculescu's formula for free convolution}
\label{subsec:voiculescu}

\noindent Here we show how our results provide a non-trivial generalisation of Voiculescu's formula for free convolution of probability distributions to a non-free situation. This is essentially an adaption of the derivations in \cite{ZJ99-2,Zee96} to the case where the ``external'' matrix follows a Gaussian distribution; gradually turning off the $O(q)$-symmetry breaking interactions of the Potts model, our formulae should reduce to Voiculescu's for free random variables. To confirm this is the case, it is convenient to consider the slightly generalised measure $\mu_\lambda$,

\begin{equation}
\label{pottsmeasure2}
\D\mu_\lambda(X_1,X_2,\dots X_q)=\frac{1}{Z_{N,q}^{\lambda}}\prod_{\langle i j\rangle}\E^{\lambda N\mathrm{tr}X_i X_j}\times \prod_{i=1}^q\E^{-N\mathrm{tr}V_i(x)}\D X_i\ ,\quad \lambda \geq 0\ ,
\end{equation}

\noindent which reduces to $\mu$ for $\lambda\to1$ (cf. \eref{pottsmeasure}) and should yield Voiculescu's formula for $\lambda\to 0$. Of course, the parameter $\lambda$ is redundant in that we may equivalently obtain $\mu$
by a suitable rescaling of $X_i$ and $\{t_m\}_{m=2}^{k+2}$; we are thus not departing from the initial parameter space of the model. 
The following holds for averages with respect to $\mu_\lambda$.

\begin{prop} 
Take $N\to\infty$. Then as $\lambda\to 0$, 
\begin{eqnarray}
\label{freelimit1}
G^M_{\lambda Y}(z)-W_Y(z)&\to R^M(z)\ ,\\
\label{freelimit2}
W_{(q)}^{-1}(z)&\to \sum_{i=1}^qW_{X_i}^{-1}(z)-\frac{q-1}{z}\ ,
\end{eqnarray}

\noindent where $R^M(z)$ denotes the R-transform of $W_M(z)$ as defined in \eref{rtrafo}. 
\end{prop}

\begin{proof} 
According to Lemma \ref{lem-Z}, we can write the partition function $Z_{N,q}^{\lambda}$
via the fiducial matrix $Y=\sqrt{1-\E^{-2h}}P_+$ as

\begin{eqnarray}
\label{Zalpha1}
Z_{N,q}^{\lambda}&=\int_\mathbb{R}\mathrm{d}Y\ \E^{-N\mathrm{tr}Y^2/2}\prod_{i=1}^q\int_\Gamma\mathrm{d}X_i\ \E^{-N\mathrm{tr}[V_i(X_i)+X_i^2/2-\lambda X_iY]}\\
\label{Zalpha2}
&\equiv \int_\mathbb{R}\mathrm{d}Y\ \E^{-N\mathrm{tr}Y^2/2}\left(\prod_{i=1}^q\int_\Gamma\mathrm{d}X_i\ \E^{-N\mathrm{tr}[V_i(X_i)+X_i^2/2]}\right)\E^{\lambda N\mathrm{tr}Y\sum_{i=1}^q X_i}\ .
\end{eqnarray}

\noindent Diagonalising the matrices and integrating over the unitary group, we can write, taking the limit $N\to\infty$,

\begin{equation}
\frac{1}{N}\frac{\partial}{\partial z}\ln\left.\int_\Gamma\mathrm{d}X_i\ \E^{-N\mathrm{tr}[V_i(X_i)+X_i^2/2-\lambda X_iY]}\right|_{y_N=z}=G_{\lambda Y}^{X_i}(z)-W_{\lambda Y}(z)\ ,
\end{equation}

\noindent where we used the definition \eref{def-g}; comparing \eref{Zalpha1} and \eref{Zalpha2}, this implies

\begin{equation}
\label{Voicu-gen}
G_{\lambda Y}^{X_1+X_2+\dots X_q}(z)-W_{\lambda Y}(z)=\sum_{i=1}^q\left(G_{\lambda Y}^{X_i}(z)-W_{\lambda Y}(z)\right)\ .
\end{equation}

\noindent Now consider the limit $\lambda\to 0$. On the one hand,

\begin{equation}
1=\lim_{\lambda\to 0}\frac{\E^{\lambda Nx_i y_j}}{\Delta(x)\Delta(y)}\ ,
\end{equation} 

\noindent from which it follows that 

\begin{eqnarray}
0&=\lim_{\lambda\to 0}\left(G_{X_i}^{\lambda Y}(z)-W_{X_i}(z)\right)\ ,\\
0&=\lim_{\lambda\to 0}\left(G_{X_1+X_2+\dots X_q}^{\lambda Y}(z)-W_{(q)}(z)\right)\ .
\end{eqnarray} 

\noindent On the other hand, as can be seen from \eref{SPE3} in this limit, the matrix $Y$ will follow a gaussian distribution, so its spectral density approaches a semi-circle. As a result, the spectral density of the rescaled matrix $\lambda Y$ approaches a delta function so that $W_{\lambda Y}(z)\to 1/z$. Together with the relation $G_X^Y\circ G_Y^X=\mathrm{id}$, this means that indeed 

\begin{equation}
G^M_{\lambda Y}(z)-W_{\lambda Y}(z)\to R^M(z)\quad\mathrm{as}\ \lambda\to 0\ ,
\end{equation}

\noindent from comparison with the definition \eref{rtrafo}. Lastly, inserting the above into \eref{Voicu-gen} yields \eref{freelimit2}.
\end{proof}

\noindent For $q=2$, \eref{freelimit2} indeed gives Voiculescu's formula \eref{Voicu}. It is in this sense that the function $G^M_{\lambda Y}(z)-W_{\lambda Y}(z)$ lifts the notion of the R-transform, so that \eref{Voicu-gen} represents a nontrivial extension of Voiculescu's formula to the addition of \em correlated \em random matrices, distributed according to $\mu_\lambda$.
It is instructive to compare \eref{Voicu-gen} for $\lambda=1$ to the expressions in Proposition \ref{prop-curve} of the previous section more explicitly. Since from \eref{Zalpha1} and \eref{Zalpha2}

\begin{eqnarray}
G_Y^{X_i}(z)&=\frac{q-1}{q}W_Y(z)_++\frac{1}{q}\left(z-W_Y(z)_-\right)\ ,\\ G_Y^{X_1+\dots X_q}(z)&=z-W_Y(z)_-\ ,
\end{eqnarray}

\noindent we observe upon comparison to \eref{curve2} that indeed

\begin{equation}
\label{identificationG}
G_Y^{(1)}(z)=G_Y^{X_i}(z)\ ,\quad G_Y^{(q)}(z)=G_Y^{X_1+\dots X_q}(z)\ .
\end{equation}

\noindent Hence, for the $S_q$-invariant case 
$V_i(z)\equiv U(z)-z^2/2\ \forall i$, our main result in Proposition \ref{prop-curve} \em further \em generalises this result to the sum of $p\leq q$ matrices: the function $G^{(p)}_Y(z)-W_Y(z)$ generalises the R-transform of $W_{(p)}(z)$, and \Eref{curve2} generalises Voiculescu's formula. 

\section{Case studies}
\label{sec:cases}

\noindent Here we consider the cases $(q,k)=(1,2)$, $(2,1)$, and $(3,1)$, which describe hard dimers, the $A_3$ and the $D_4$ model on planar triangulations, respectively. For the former two models, the functions $W_{(p)}(z)$ have been known for a while \cite{Staudacher89,Carroll96,Carroll97} -- the fact that our general formula in Proposition \ref{prop-curve} reproduces these results lends credence to our extension to the $D_4$ model. Unlike models with irrational values of $\arccos((q-2)/2)/\pi$, all of these share the simplification that they can be described by polynomial equations: We derive explicit expressions for the polynomials $F_{(p)}(x,y)$ satisfying 

\begin{equation}
\label{pol}
F_{(p)}\left(z,G_{(p)}^Y(z)\right)=0\ ,\quad 1\leq p\leq q\ ,
\end{equation}

\noindent which define a family of algebraic curves $\mathcal{C}_{(p)}=\{(x,y)\in\mathbb{C}^2|F_{(p)}(x,y)=0\}$. Here we exclusively consider solutions for which the spectral densities have connected support such that the constants $c^{(p)}_{i,j}$ that appear in the expressions for $F_{(p)}(x,y)$ below may be determined as functions of $\{t_m\}_{m=2}^{k+2}$ via the requirement that $\mathcal{C}_{(p)}$ be of genus zero. In \ref{app:asympt}, we describe the resulting analytic structure of $G_Y^{(p)}(z)$ and $G_{(p)}^Y(z)$. 

\subsection{$(q,k)=(1,2)$ -- Hard dimers}
\label{subsec:q=1}

\noindent This model describes hard dimers on planar triangulations and was first solved on the sphere by Staudacher \cite{Staudacher89}. According to \eref{pottsmeasure} and \eref{Z1}, the partition function can be written as both a one- and two-matrix integral,

\begin{eqnarray}
Z_{N,1}&=\int\mathrm{d}X\ \E^{-N\mathrm{tr}[U(X)-X^2/2]}\\
\label{Z2-2}
	&=\int\mathrm{d}Y\ \E^{-N\mathrm{tr}Y^2/2}\int\mathrm{d}X\ \E^{-N\mathrm{tr}[U(X)-XY]}\ .
\end{eqnarray}

\noindent Using the definition \eref{def-g} in the planar limit, the above expressions imply the following relations: 

\begin{eqnarray}
\label{case1a}
z&=W_Y(z)_-+G_Y^X(z)_+\ ,\\
\label{case1b}
U'(z)&=W_{(1)}(z)_-+G_X^Y(z)_+\ ,\\
\label{case1c}
U'(z)&=W_{(1)}(z)_-+W_{(1)}(z)_++z\ ,
\end{eqnarray}

\noindent The first line above is indeed consistent with \eref{curve2} in Proposition \ref{prop-curve} and \eref{identificationG} in Subsection \ref{subsec:voiculescu}. Through the relation $G_Y^X\circ G_X^Y=\mathrm{id}$, \eref{case1a} and \eref{case1b} dictate the analytic structure and asymptotic behaviour of $G_Y^X(z)$, the result of which is given in detail in \ref{ex-dimers}. This allows us to compute the spectral curve using \eref{curve2},

\begin{eqnarray}
\label{polynomial0}
F_{(1)}(x,y)&=x^4-x^3y+\frac{t_3}{t_4}x^3+\frac{y^2}{t_4}-\frac{t_3}{t_4}x^2y+\frac{t_2+t_4}{t_4}x^2\nonumber\\&\quad-\frac{t_2+1}{t_4}xy-c_{0,0}x+ c_{1,1}y+c_{1,0}\ .
\end{eqnarray}

\noindent According to Corollary \ref{curve3}, the functions $G_{(p)}^Y(z)$ then satisfy

\begin{equation}
F_{(1)}\left(G_{(0)}^Y(z)-z,G_{(0)}^Y(z)\right)=0\ ,\quad F_{(1)}\left(z,G_{(1)}^Y(z)\right)=0\ ,
\end{equation} 

\noindent which in turn determines their analytic structure and asymptotic behaviour on all sheets -- see \ref{ex-dimers}. Finally, comparing to \eref{case1c}, we conclude that 

\begin{equation} 
\label{G-1}
G_{(1)}^Y(z)_+=z+W_{(1)}(z)\ ,\quad G_{(1)}^Y(z)_-=t_4z^3+t_3z^2+t_2z-W_{(1)}(z)\ .
\end{equation}

\subsection{$(q,k)=(2,1)$ -- Ising model}
\label{subsec:q=2}

\begin{figure}[h]
\centering
\input{statesising}
	\caption{Integrable boundary conditions for the Ising model ($q=2$) on a flat lattice are labelled by the nodes of the graph $A_2\times A_3$; the dashed line separates two equivalent choices of a fundamental domain.}
	\label{fig:statesising}
\end{figure}
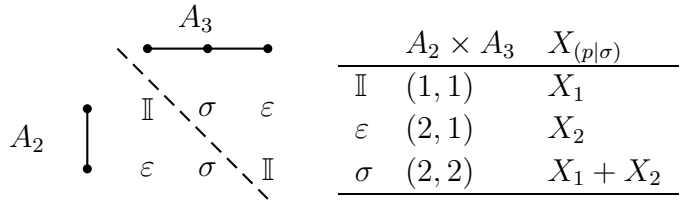

\noindent This corresponds to the Ising model on planar triangulations, which is the $A_3$ model in the classification of \cite{Pasquier86} and was first solved on the sphere by Kazakov and Boulatov \cite{Boulatov86,Kazakov86}; it is described by the much-studied symmetric hermitian two-matrix model. The three integrable boundary conditions of the model are captured by the linear combinations of $X_i$ shown in Figure \ref{fig:statesising} \cite{Cardy89}: $W_{(1)}(z)$ captures the $S_2\simeq\mathbb{Z}_2$-doublet $\{\mathbb{I},\varepsilon\}$, $W_{(2)}(z)$ the $\mathbb{Z}_2$-singlet $\{\sigma\}$. From \eref{Z1}, we see that in this case the partition function can be written as

\begin{eqnarray}
\label{Z2-1}
Z_{N,2}&=\int\mathrm{d}X_1\mathrm{d}X_2\ \E^{-N\mathrm{tr}[U(X_1)+U(X_2)]}\E^{N\mathrm{tr}(X_1+X_2)^2/2}\\
\label{Z2-2}
	&=\int\mathrm{d}Y\ \E^{-N\mathrm{tr}Y^2/2}\left(\int\mathrm{d}X\ \E^{-N\mathrm{tr}[U(X)-XY]}\right)^2\ .
\end{eqnarray}

\noindent On the other hand, changing variables to $X_\pm=X_1\pm X_2+t_2/t_3$ and integrating out $X_-$, we obtain the equivalent one-matrix representation going back to \cite{Eynard92},

\begin{eqnarray}
Z_{N,2}
&=\mathrm{const.}\times\int\frac{\mathrm{d}X_+\ \E^{-N\mathrm{tr}U_+(X_+)}}{\sqrt{\mathrm{Det}(X_+\otimes\mathbb{I}+\mathbb{I}\otimes X_+)}}\ ,
\end{eqnarray}

\noindent where $U'_+(z)=t_3z^2/4-z-t_2(4-t_2)/(4t_3)$ and capital $\mathrm{Det}$ denotes the determinant on $N^2\times N^2$ matrices. Using the definition \eref{def-g} in the planar limit, the above expressions respectively imply the following set of equations:

\begin{eqnarray}
\label{case2a}
z&=W_Y(z)_--W_Y(z)_++2G_Y^X(z)_+\ ,\\
\label{case2b}
U'(z)&=W_{(1)}(z)_-+G_X^Y(z)_+\ ,\\
\label{case2c}
U'(z)&=W_{(1)}(z)_-+G_{X_1}^{X_2}(z)_++z\ ,\\
\label{case2d}
U'_+\left(z+t_2/t_3\right)&=W_{(2)}(z)_-+W_{(2)}(z)_++W_{(2)}(-z)\ .
\end{eqnarray}

\noindent Again, the first line is consistent with \eref{curve2} in Proposition \ref{prop-curve} and \eref{identificationG} in Subsection \ref{subsec:voiculescu}. Equations \eref{case2a} and \eref{case2b} dictate the analytic structure and asymptotic behaviour of $G_Y^X(z)$, cf. \ref{ex-ising}. As before, this allows us to compute the spectral curve using \eref{curve2},

\begin{eqnarray}
\label{polynomial1}
F_{(1)}(x,y)&=x^4-2x^3y-\frac{1}{t_3}y^3+\frac{1-t_2}{t_3^2}y^2+x^2y^2-\frac{t_2+2}{t_3}x^2y\nonumber\\
&\quad+\frac{t_3^2-t_2^2}{t_3^2}x^2+\frac{t_2+2}{t_3}xy^2+\frac{t_2^2-t_3^2}{t_3^2}xy+c^{(1)}_{1,1}x+c^{(1)}_{1,0}\ ,\\
F_{(2)}(x,y)&=x^4+\frac{4t_2}{t_3}x^3+\frac{4}{t_3}y^3-x^2y^2-\frac{4+2t_2}{t_3}x^2y-\frac{2 t_2}{t_3}xy^2+\frac{4t_2^2+2t_3^2}{t_3^2}x^2\nonumber\\&\quad+\frac{8t_2}{t_3^2}y^2-\frac{4t_2(2+t_2)}{t_3^2}xy-c^{(2)}_{0,0}x+c^{(2)}_{1,1}y+c^{(2)}_{1,0}\ .
\end{eqnarray}

\noindent According to Corollary \ref{curve3}, the functions $G_{(p)}^Y(z)$ then satisfy

\begin{eqnarray} 
&F_{(1)}\left(z,G_{(1)}^Y(z)\right)&=0\ ,\quad F_{(1)}\left(G_{(1)}^Y(z)-z,G_{(1)}^Y(z)\right)=0\ ,\\
&F_{(2)}\left(z,G_{(2)}^Y(z)\right)&=0\ ,\quad F_{(2)}\left(G_{(0)}^Y(z)-z,G_{(0)}^Y(z)\right)=0\ .
\end{eqnarray} 

\noindent Again we may use the above to compute the analytic structure and asymptotic behaviour of $G_{(p)}^Y(z)$ on all sheets -- see \ref{ex-ising}. Comparing to \eref{case2c} and \eref{case2d}, we conclude that for $p=1$

\begin{eqnarray}
\label{G-2a}
G_{(1)}^Y(z)_+&=z+G_{X_1}^{X_2}(z)\ ,\\G_{(1)}^Y(z)_-&=t_3z^2+t_2z-W_{(1)}(z)\ ,
\end{eqnarray}

\noindent whereas for $p=2$

\begin{eqnarray}
\label{G-2b}
G_{(2)}^Y(z)_+&=z+W_{(2)}(z)\ ,\\G_{(2)}^Y(z)_-&=t_3z^2/4+t_2z/2-W_{(2)}(z)-W_{(2)}(-z)\ .
\end{eqnarray}

\noindent Our results reproduce the analytic structure found in \cite{Carroll96,Carroll97,Atkin11} as well as the relation between the $p=1$ and $p=2$ boundary conditions reported in \cite{Atkin12}: at the level of the polynomial equation, the correspondence with the quantities defined therein is

\begin{equation}
W_Y(z) \leftrightarrow W_A(a)\ ,\quad G_Y^{(1)}(z) \leftrightarrow x(a)\ ,\quad G_Y^{(2)}(z) \leftrightarrow m(a)\ .
\end{equation} 

\noindent The polynomial $E(x,y)=-t_3F_{(1)}(x,x+y)$ is of order 3 in both $x$ and $y$, 

\begin{eqnarray}
E(x,y)&=x^3+y^3-t_3x^2y^2-\frac{1-t_2}{t_3}(x^2y+y^2x)-\frac{1-t_2}{t_3}(x^2+y^2)\nonumber\\
&\quad-\frac{2-2t_2+t_2^2-t_3^2}{t_3}xy-t_3c^{(1)}_{1,1}(x+y)-t_3c^{(1)}_{1,0}\ ,
\end{eqnarray}

\noindent and satisfies $E(x,y)=E(y,x)$ and $E(z,G_{X_1}^{X_2}(z))=0$, as follows from comparison of \eref{case2a} and \eref{case2b}. This is the usual spectral curve of the two-matrix model introduced by Eynard \cite{Eynard02}.

\subsection{$(q,k)=(3,1)$ -- 3-states Potts model}
\label{subsec:q=3}

\begin{figure}[h]
\centering
	\input{statespotts}
	\caption{Integrable boundary conditions for the 3-states-Potts model ($q=3$) on a flat lattice are labelled by the nodes of the graph $A_4\times D_4$; the dashed line separates two equivalent choices of a fundamental domain.}
	\label{fig:statespotts}
\end{figure}
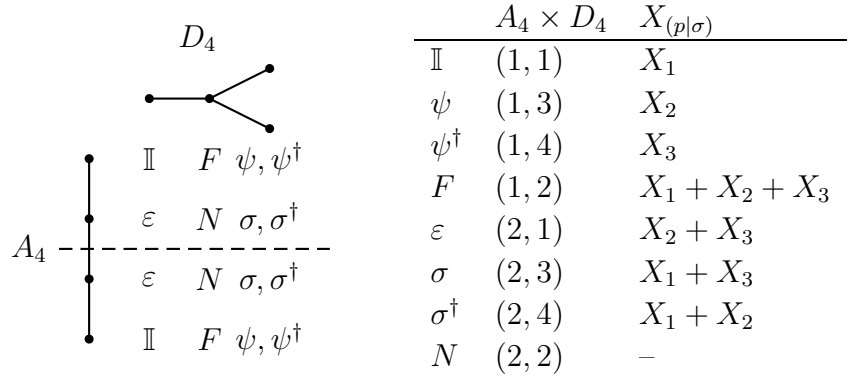

\noindent This model is equivalent to the $D_4$ lattice model on planar triangulations, for which $W_{(1)}(z)$ was first calculated by Daul in \cite{Daul94}. The full list of boundary conditions of the $D_4$ lattice model is given in Figure \ref{fig:statespotts} -- see \cite[p.60]{Behrend00}: $W_{(1)}(z)$ captures the $S_3$-triplet $\{\mathbb{I},\psi,\psi^{\dagger}\}$, $W_{(2)}(z)$ the $S_3$-triplet $\{\varepsilon,\sigma,\sigma^{\dagger}\}$, and $W_{(3)}(z)$ the singlet $\{F\}$; thanks to Corollary \ref{curve3}, the spectral curve for the latter also defines another singlet $W_{(0)}(z)$, which may be conjectured to describe the one remaining independent boundary condition $\{N\}$, though herein we will not attempt to determine its relationship to the microscopic definition given in \cite{Behrend00}. From \eref{Z1}, we see that the partition function can be written as

\begin{eqnarray}
\label{Z-potts-3}
Z_{N,3}&=\int\mathrm{d}X_1\mathrm{d}X_2\mathrm{d}X_3\ \E^{-N\mathrm{tr}[U(X_1)+U(X_2)+U(X_3)]}\E^{N\mathrm{tr}(X_1+X_2+X_2)^2/2}\\
	&=\int\mathrm{d}Y\ \E^{-N\mathrm{tr}Y^2/2}\left(\int\mathrm{d}X\ \E^{-N\mathrm{tr}[U(X)-XY]}\right)^3\ .
\end{eqnarray}

\noindent Again we may set $X_\pm=X_1\pm X_2+t_2/t_3$ and integrate out $X_-$, which gives

\begin{eqnarray}
\label{lightcone}
Z_{N,3}
&=\mathrm{const.}\times\int\frac{\mathrm{d}X_+\mathrm{d}X_3\ \E^{-N\mathrm{tr}[U_+(X_+)+U(X_3)-X_+X_3]}}{\sqrt{\mathrm{Det}(X_+\otimes\mathbb{I}+\mathbb{I}\otimes X_+)}}\ ,
\end{eqnarray}

\noindent where $U_+(z)$ is as in the previous section. Using the definition \eref{def-g} in the planar limit, the above expressions respectively imply the following set of equations:

\begin{eqnarray}
\label{case3a}
z&=W_Y(z)_--2W_Y(z)_++3G_Y^X(z)_+\ ,\\
\label{case3b}
U'(z)&=W_{(1)}(z)_-+G_X^Y(z)_+\ ,\\
\label{case3c}
U'(z)&=W_{(1)}(z)_-+G_{X_3}^{X_1+X_2}(z)_++z\ ,\\
\label{case3d}
U'_+\left(z+t_2/t_3\right)&=W_{(2)}(z)_-+G_{X_1+X_2}^{X_3}(z)_++W_{(2)}(-z)\ .
\end{eqnarray}

\noindent Once again, the first line is consistent with \eref{curve2} in Proposition \ref{prop-curve} and \eref{identificationG} in Subsection \ref{subsec:voiculescu}. The analytic structure and asymptotic behaviour of all relevant functions can be determined as before - cf. \ref{ex-potts3}. The resulting spectral curves are

\begin{eqnarray}
\label{polynomial2}
F_{(1)}(x,y)&=x^6-x^5 \left(6 y+\frac{6 t_2}{t_3}\right)+x^4 \left(13 y^2+\frac{6 (4 t_2-1)}{t_3}y+\frac{9 t_2^2}{t_3^2}+2\right)\nonumber\\
&-x^3 \left(12 y^3-\frac{24-28 t_2}{t_3}y^2+\frac{12 t_2^2+8 t_3^2-24t_2}{t_3^2}y-\frac{4t_2^3-8 t_2 t_3^2}{t_3^3}\right)\nonumber\\
&+x^2 \left(4 y^4+\frac{6 (t_2-6)}{t_3}y^3+\frac{9-54 t_2-15 t_2^2+10 t_3^2}{t_3^2}y^2\right.\nonumber\\
&\qquad\left. -\frac{6\left(3 t_2^2+5 t_2^3+t_3^2-3 t_2 t_3^2\right)}{t_3^3}y-\frac{12 t_2^4}{t_3^4}+\frac{6 t_2^2}{t_3^2}+1\right)\nonumber\\
&+x \left(\frac{4
\left(3+t_2\right)}{t_3}y^4+\frac{2\left(12 t_2+9 t_2^2-2 t_3^2-9\right)}{t_3^2}y^3\right. \nonumber\\
&\qquad\left.+\frac{2\left(13 t_2^3+6 t_3^2-t_2 \left(9+4 t_3^2\right)\right)}{t_3^3}y^2\right. \nonumber\\
&\qquad\left.-\frac{2\left(6 t_2^3-6
t_2^4-6 t_2 t_3^2+t_3^4\right)}{t_3^4}y+\frac{4 t_2^3-2 t_2 t_3^2}{t_3^3}\right)\nonumber\\
&-\frac{4 }{t_3}y^5+\frac{17-18
t_2}{t_3^2}y^4+c^{(1)}_{3,3}y^3
+c^{(1)}_{3,2}y^2+c^{(1)}_{3,1}y+c^{(1)}_{3,0}\ ,
\end{eqnarray}

\begin{eqnarray}
\label{polynomial3}
F_{(3)}(x,y)&= x^6+x^5 \left(6 y+\frac{18 t_2}{t_3}\right)+x^4 \left(9 y^2+\frac{18\left(4 t_2-1\right)}{t_3}y+\frac{117 t_2^2}{t_3^2}+6\right)\nonumber\\
&-x^3 \left(4 y^3-\frac{36\left(t_2-2\right)}{t_3}y^2-\frac{12
\left(21 t_2^2+2 t_3^2-18t_2\right)}{t_3^2}y-\frac{36 \left(9 t_2^3+2 t_2 t_3^2\right)}{t_3^3}\right)\nonumber\\
&-x^2 \left(12 y^4+\frac{18 \left(3+5 t_2\right)}{t_3}y^3-\frac{9\left(9-54 t_2-15 t_2^2+2 t_3^2\right)}{t_3^2}y^2\right.\nonumber\\
&\qquad\left.-\frac{54\left(-15 t_2^2+3 t_2^3-t_3^2+3
t_2 t_3^2\right)}{t_3^3}y-\frac{270
t_2^2}{t_3^2}-\frac{324 t_2^4}{t_3^4}-9\right)\nonumber\\
&-x \left(\frac{36 \left(t_2-1\right)}{t_3}y^4+\frac{6\left(39
t_2^2+2 t_3^2-27\right)}{t_3^2}y^3\right.\nonumber\\
&\qquad\left.+\frac{54\left(12 t_2^2+9 t_2^3+2 t_3^2-9 t_2\right)}{t_3^3}y^2-\frac{54 t_2 \left(6 t_2^2+t_3^2\right)}{t_3^3}\right.\nonumber\\
&\qquad\left. -\frac{18\left(12
t_2^2 t_3^2+t_3^4-54 t_2^3-18 t_2^4-18 t_2 t_3^2\right)}{t_3^4}y\right)\nonumber\\
&+\frac{108}{t_3} y^5+\frac{27  \left(26 t_2-9\right)}{t_3^2}y^4-\frac{9 \left(12-84
t_2^2+16 t_3^2+8 t_2 \left(t_3^2-9\right)\right)}{t_3^3}y^3 \nonumber\\&+\frac{3
\left(216 t_2^3+36 t_3^2-144 t_2 t_3^2+t_2^2 \left(216-84 t_3^2\right)+8t_3^4\right)}{t_3^4}y^2\nonumber\\&-\frac{9\left(48 t_2^2+24 t_2^3-8 t_2 t_3^2+4t_3^2\right)}{t_3^3}y+\frac{72
t_2^2}{t_3^2}+4 \nonumber\\&-27 \left(c^{(1)}_{3,3}y^3+c^{(1)}_{3,2}y^2+c^{(1)}_{3,1}y+c^{(1)}_{3,0}\right)\ .
\end{eqnarray}

\noindent According to Corollary \ref{curve3}, the functions $G_{(p)}^Y(z)$ then satisfy

\begin{eqnarray} 
&F_{(1)}\left(z,G_{(1)}^Y(z)\right)&=0\ , \quad F_{(1)}\left(G_{(2)}^Y(z)-z,G_{(2)}^Y(z)\right)=0\ ,\\
&F_{(3)}\left(z,G_{(3)}^Y(z)\right)&=0\ ,\quad F_{(3)}\left(G_{(0)}^Y(z)-z,G_{(0)}^Y(z)\right)=0\ .
\end{eqnarray} 

\noindent As before, the above fixes the analytic structure and asymptotic behaviour of $G_{(p)}^Y(z)$ on all sheets -- see \ref{ex-potts3}. Comparing to \eref{case3c} and \eref{case3d}, we conclude that 

\begin{eqnarray}
\label{G-3a}
G_{(1)}^Y(z)_+&=z+G^{X_1+X_2}_{X_3}(z)\ ,\\ G_{(1)}^Y(z)_-&=t_3z^2+t_2z-W_{(1)}(z)\ ,
\end{eqnarray}

\noindent and

\begin{eqnarray}
\label{G-3b}
G_{(2)}^Y(z)_+&=z+G^{X_3}_{X_1+X_2}(z)\ ,\\ G_{(2)}^Y(z)_-&=t_3z^2/4+t_2z/2-W_{(2)}(z)-W_{(2)}(-z)\ .
\end{eqnarray}

\noindent Similarly, one can show 

\begin{eqnarray}
\label{G-3c}
G_{(3)}^Y(z)_+&=z+W_{(3)}(z)\ ,\\ G_{(3)}^Y(z)_-&=t_3z^2/9+t_2z/3+\mathcal{O}(z^{-1})\ .
\end{eqnarray}

\noindent $F_{(1)}(x,y)$ corresponds to the spectral curve first described in \cite{Daul94,ZJ99}; the remaining expressions are new results. The polynomials $-t_3F_{(p)}(x,x+y)=4E_{(p)}(x,y)$ are of degree one less in $x$. For example,

\begin{eqnarray}
\label{e1}
E_{(1)}(x,y)&=x^5+y^5-x^4 \left(\frac{t_3}{4} y^2+\frac{2 t_2 -20}{4 }y+\frac{8-24 t_2}{4 t_3}\right)-\frac{17-18 t_2}{4 t_3}y^4\nonumber\\&\quad-x^3 \left(t_3 y^3+\frac{14 t_2 -34}{4 t_3}y^2+\frac{12 t_2^2-84 t_2+32}{4 t_3}y-\frac{t_3\widetilde{c}_1}{4}\right)\nonumber\\&\quad-x^2 \left(t_3 y^4+\frac{11\left(t_2 -1\right)}{2}y^3+\frac{39 t_2^2-90 t_2-2 t_3^2+57}{4 t_3}y^2-\frac{t_3\widetilde{c}_2}{4 }y+\frac{\widetilde{c}_3}{4}\right)\nonumber\\&\quad+x \left(\frac{8-4 t_2 t_3}{4}y^4-\frac{9t_2^2-24 t_2-2 t_3^2+25}{2 t_3}y^3+\frac{ t_3\left(3  c^{(1)}_{2,2}- c^{(1)}_{3,3}\right)}{4}y^2\right)\nonumber\\&\quad-\frac{t_3 \left(2 c^{(1)}_{3,2}- c^{(1)}_{2,1}\right)}{4}xy-\frac{1}{4} t_3 y^3 c^{(1)}_{3,3}-\frac{1}{4} t_3 y^2 c^{(1)}_{3,2}-\frac{t_3(c^{(1)}_{3,1}-c^{(1)}_{2,0})}{4}x\nonumber\\&\quad-\frac{1}{4} t_3 y c^{(1)}_{3,1}-\frac{1}{4} t_3 c^{(1)}_{3,0}\ ,
\end{eqnarray}

\noindent where 

\begin{eqnarray}
\widetilde{c}_1&=c^{(1)}_{2,2}-c^{(1)}_{1,1}-c^{(1)}_{3,3}+c^{(1)}_{0,0}\ ,\\
\widetilde{c}_2&=c^{(1)}_{2,2}-2c^{(1)}_{1,1}-3c^{(1)}_{3,3}\ ,\\
\widetilde{c}_3&=c^{(1)}_{1,0}-c^{(1)}_{2,1}+c^{(1)}_{3,2}\ .
\end{eqnarray}

\noindent This expression satisfies $E_{(2)}(x,y)=E_{(1)}(y,x)$ and is equivalent to the polynomial $Q(x_3,x_+)$ reported previously by the authors in \cite{ksm13}. Upon inspection of \eref{case3b}, \eref{case3c} and \eref{case3d} we conclude that 

\begin{equation}
0=E_{(2)}\left(z,G_{X_1+X_2}^{X_3}(z)\right)=E_{(2)}\left(G_{X_3}^{X_1+X_2}(z),z\right)\ .
\end{equation}

\section{Critical behaviour}
\label{sec:crit}

\noindent This section discusses the critical behaviour of $W_{(p)}(z)$ for $0< q< 4$. The existence of a second-order phase transition for the Potts model in this regime has been demonstrated on a flat lattice by Baxter \cite{Baxter71,Baxter73}; here we describe their random-lattice counterparts\footnote{When $q>4$, these critical points do not exist, though presumably another critical point emerges as for the $O(n)$ model on planar triangulations, for which $\gamma_s=1/2$ when $n>2$ \cite{Eynard95,Durhuus96}.}. According to Proposition \ref{prop-curve}, it suffices to determine the critical behaviour of $W_Y(z)$ for $1>\nu>0$. Its possible critical exponents are determined by the multiplicity of the singularity at the left edge $z_-$ of the spectral density $\rho_Y(z)$, which controls the large-order behaviour of the generating function $W_Y(z)$. 

Let us consider the case of triangulations covered in Proposition \ref{prop-ycurve}. Then $z_-=\delta_U$ when both $t_2$ and $t_3$ are at their critical values $t_{2,c}$, $t_{3,c}$, with $t_{m>3}=0$. When $\nu$ is rational, exact expressions for the critical lines and points can be obtained easily by requiring sufficiently many derivatives of the polynomial $F_{(p)}(x,y)$ 
to vanish; the result is depicted for the example $(q,k)=(3,1)$ in Figure \ref{fig:phasediagram}. For example, from \eref{polynomial0}, \eref{polynomial1} and \eref{polynomial2} respectively, we find

\begin{figure}
\label{fig:phasediagram}
\centering
\includegraphics[width=.5\textwidth]{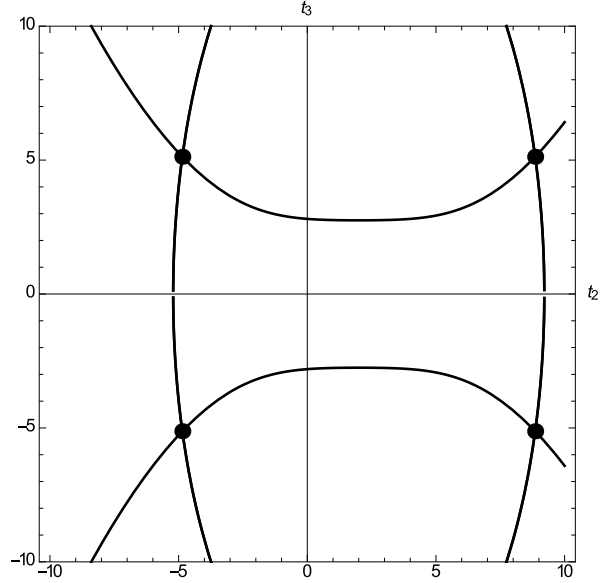}
\caption{A portion of the phase diagram of the 3-states Potts model on planar triangulations. Along the critical lines, $\partial_y^4 E_{(1)}$, $\partial_x\partial_y^3 E_{(1)}$ and $\partial_x^3\partial_y E_{(1)}$ vanish, with the polynomial $E_{(1)}$ as in \Eref{e1}; at the critical points, $\partial_x^4 E_{(1)}$ vanishes in addition.}
\end{figure}

\begin{equation}
(t_{2,c},t_{3,c})=\left\{\begin{array}{ll}
\left(1\pm2\sqrt{3},\pm\sqrt{2}\right) \ ,\quad& q=1\ ,\\
\left(2\pm2\sqrt{7},\pm\sqrt{10}\right)\ ,\quad& q=2\ ,\\
\left(3\pm\sqrt{47},\pm\sqrt{105}/2\right)\ ,\quad& q=3\ .\end{array}\right. 
\end{equation}

 \noindent Let us parametrise the vicinity of this point by eliminating $\delta_U$ in favour of the scaling parameter $\varepsilon=z_--\delta_U$ and investigate the limit $\varepsilon\to 0$. We would like to expand $W_Y(z)$ in powers of $\varepsilon$, keeping $(z-z_-)/\varepsilon$ finite. Setting again $w(z)=\sqrt{z-\delta_U}$, this requires the expansion of $f(w)$ in \eref{lambdares} in powers of $\sqrt{\varepsilon}\equiv w_-$, keeping $w/\sqrt{\varepsilon}$ finite. As we show in \ref{app:auxprob}, 
equation
\eref{f-scaling}, the terms of $\mathcal{O}(\varepsilon^{n\pm\nu/2})$ in the expansion of $f_{\mathrm{s.}}(w)$ can be written as

\begin{equation}
	\label{Gscaling}
	\varepsilon^{n\pm\nu/2}\left
(t_n^{(\pm)}T_{2n\pm\nu}(-w/\sqrt{\varepsilon})+u_n^{(\pm)}U_{2n\pm\nu}(-w/\sqrt{\varepsilon})\right)\ ,
\end{equation}

\noindent where $T_\nu(w)$ (resp. $U_\nu(w)$) is the Chebyshev function of the first (resp. second) kind as defined in \eref{chebyshevfct}. Using \eref{Chebyshevdiscont} to compare the term of same order in the expansion of the discontinuity $f(w)_+-f(w)_-$ across $w/\sqrt{\varepsilon}\in[1,\infty)$ to \eref{lambdadens} and requiring that $\rho_Y(z)\to 0$ as $z\to z_-$ reveals that $u_n^{(\pm)}=0$ for all $n$. Using the relationship $(z-z_-)/\varepsilon= w^2/\varepsilon-1$ and \eref{lambdares} together with

\begin{equation}
T_{2-\nu}\left(\sqrt{1-\eta}\right)+T_{2-\nu}\left(-\sqrt{1-\eta}\right)=2\cos\left(\frac{\pi\nu}{2}\right)T_{2-\nu}\left(\sqrt{\eta}\right)\ ,
\end{equation}

\noindent gives the following expansion of $W_Y(z)$:

\begin{eqnarray}
W_Y(z_--\varepsilon \eta)&=W_Y(z_-)+C\varepsilon^{1-\nu/2}T_{2-\nu}\left(\sqrt{\eta}\right)-\varepsilon\frac{2\eta}{4-q}+\mathcal{O}(\varepsilon^{1+\nu/2})\ ,
\end{eqnarray}

\noindent where $W_{Y}(z_-)=(2z_-+qt_{2,c}/t_{3,c})/(4-q)$ and $C$ is a normalisation constant. The expansion of $G_Y^{(p)}(z)$ now follows immediately from Proposition \ref{prop-curve}; the leading non-analytic term reads

 \begin{equation}
 C\varepsilon^{1-\nu/2}\left[T_{2-\nu}\left(\sqrt{\eta}\right)-\frac{2p}{q}\cos\left(\frac{\pi\nu}{2}\right)T_{2-\nu}\left(\sqrt{1-\eta}\right)\right]\ .
 \end{equation}
 
\noindent The string exponent $\gamma_s$ predicted by $W_{(p)}(z-z_c)\sim (z-z_c)^{1-\gamma_s}$ is in agreement with previous findings \cite{Daul94,ZJ99,Eynard99}, namely

\begin{equation}
	\gamma_s=\frac{\nu}{\nu-2}\ .
\end{equation}
 
\noindent In particular, $\gamma_s=-1/2$, $-1/3$, $-1/5$ and $0$ for $q=1$, $2$, $3$ and $4$ respectively, which is consistent with Liouville theory interacting with conformal matter of central charges $c_M=0$, $1/2$, $4/5$ and $1$. Whilst in the first two cases the conformal field theory is unique, there exist two distinct modular invariants at $c_M=4/5$, corresponding to the $(A_4,A_5)$ Virasoro minimal model and the $(A_4,D_4)$ minimal model, which admits a conserved spin-3 current is diagonal under the extended $\mathcal{W}_3$-algebra \cite{Cappelli86,Caldeira03}. In light of the $S_3$-symmetry of the partition function $Z_{N,3}$
and the resulting spectrum of boundary conditions \cite{Affleck98} -- cf. Figure \ref{fig:statespotts} -- we expect our equations to describe the latter coupled to Liouville theory, not the former. 

\section{Discussion}
\label{sec:pottsdiscussion}

\noindent Let us summarise our results. Starting from the matrix integral representation of the Potts model on a random lattice in Lemma \ref{lem-Z}, we employed the saddle point approximation to express $W_{(p)}(z)$ via the $p$-independent average $W_Y(z)$ in Proposition \ref{prop-curve}. For the case of planar triangulations, Proposition \ref{prop-ycurve} provides an explicit elliptic parametrisation of the latter for arbitrary $q\neq 4$. Just as \eref{Z1} defines an analytic continuation of the partition function to the complex $q$-plane, \eref{curve2} may 
be used to define the analytic continuation of $W_{(p)}(z)$ in the complex $p$- and $q$-plane.  Furthermore, Corollary \ref{curve3} showed that $W_{(p)}(z)$ and $W_{(q-p)}(z)$ can be related algebraically -- in the case studies in Section \ref{sec:cases}, this resulted in the partition functions with $p$ and $q-p$ colours on the boundary being described by a single spectral curve defined by the zero locus of \eref{pol}. Remarkably, equations \eref{G-2a}, \eref{G-3a} and \eref{G-3b} indicate that $G_{(p)}^Y(z)-z$ and $G_{(q-p)}^Y(z)-z$ are functional inverses, generalising the well-known duality\footnote{Note that this involution is in general distinct from the Kramers-Wannier duality \cite{Kramers41} on the dynamical lattice: e.g. for $q=3$, the latter interchanges $p=1$ with $p=3$, and $p=2$ with $p=0$, mixing singlets and triplets \cite{Affleck98}.} interchanging the two matrices of the $\mathbb{Z}_2$-invariant hermitian two-matrix model \cite{BEH01a,BEH01b}. 

Our results naturally pave the way for a number of further developments: Firstly, going beyond the planar limit, as was done in \cite{eynard11b} for the $O(n)$ model on random lattices, it would be interesting to explore if and when the curves  $\mathcal{C}_{(p)}$ defined by \eref{pol} can be used as a valid part of the initial data of the topological recursion algorithm \cite{Eynard07}, which enables the computation of averages to all orders in $1/N$. Secondly, for general values of $h$ in Lemma \ref{lem-Z}, the remarkably simple result in Lemma \ref{lem-loop} should enable us to investigate the boundary renormalisation group flow relating boundary conditions with different $p$. This flow is expected to induce a partial order on the spectrum of boundary states in accordance with the boundary analogue of the $c$-theorem \cite{Zamolodchikov86}, as conjectured in \cite{Affleck91} and finally proven by Friedan and Konechny \cite{Friedan03}; it would be interesting to derive this fact directly from the matrix model, thus extending the work of \cite{Carroll96,Carroll97,Atkin12}.

Finally, it would be instructive to check if the universal results of Section \ref{sec:crit} can be reproduced by other means, e.g. by explicitly constructing the corresponding conformal field theory. Remarkably, as exemplified by the case of the $D_4$ model, this appears to require a non-diagonal partition function in the Liouville sector in general -- see also \cite{Wilshin08}. From this perspective, various other corners of the $(q,k)$-parameter space also warrant more detailed investigations. Of particular interest would be the computation of the scaling behaviour for strongly coupled models with $q>4$: given the close relationship between the Potts model and loop models, one might wonder if there exist analogues of the critical points of the $O(n)$ model on a random lattice with $n>2$ reported in \cite{Eynard95,Durhuus96}. 

\ack 

MA was supported by the European Research Council under the European
Union's Seventh Framework Programme (FP/2007/2013)/ ERC Grant Agreement
n. 307074, by the Belgian Interuniversity Attraction Pole P07/18, and by F.R.S.-F.N.R.S. BN thanks the mathematical physics group at Universit\'e catholique de Louvain for the kind hospitality during an intermediate stage of this project. The work of BN is supported by the German National Academic Foundation and STFC grant ST/J500641/1. JFW acknowledges the support of STFC grant ST/L000474/1.

\appendix
\section{Auxiliary Saddle Point Problem}
\label{app:auxprob}

\noindent Given $|\nu|\leq 1$, consider a function $f(w)$ holomorphic on $\mathbb{C}\setminus[\alpha,\beta]$ for some connected $[\alpha,\beta]\subset\mathbb{R}$, satisfying

\begin{equation}
	\label{linear}
	\mathrm{Re}\ f(w)=\cos (\pi \nu)f(-w)\ ,\quad w\in[\alpha,\beta]\ .
\end{equation} 

\noindent The general solution to this equation was first described in \cite{Eynard95} and we will derive it below; thereafter we investigate the limit $\alpha/\beta\to 0$.
\\

\subsection{General solution} We begin by showing that any function satisfying \eref{linear} is uniquely specified by the behaviour at its singularities. To this end, it is useful to introduce a new coordinate $\sigma$ by

\begin{eqnarray}
	\label{defnphi}
		\mathbb{C}\setminus[\pm\alpha,\pm\beta]&\longrightarrow(0,1)\times[0,\tau)\subset\mathbb{C}\;,\nonumber\\
		w&\longmapsto\sigma(w)=\frac{1}{2K}\int_{1}^{w/\alpha}\mathrm{d}t\left((1-t^2)(1-(\alpha t/\beta)^2)\right)^{-1/2}\ .
	\end{eqnarray}

\noindent By definition of the Jacobi elliptic function\footnote{Our conventions for elliptic functions are those of Gradshtein and Ryzhik \cite{Gradshteyn07}.} $\mathrm{sn}(u|k)$ of elliptic modulus $k$, the inverse map is

\begin{eqnarray}
	\label{ellipticparam}
		(0,1)\times[0,\tau)&\longrightarrow\mathbb{C}\setminus[\pm\alpha,\pm\beta]\;,\nonumber\\
		\sigma&\longmapsto w(\sigma)
		=\alpha\ \mathrm{sn}\left(2K\sigma+K\right|\alpha/\beta)\ .
	\end{eqnarray}

\noindent Here, $K$ and $K'$ are given by the complete elliptic integrals of the first and second kind, respectively:

\begin{eqnarray} 
\label{ellipticK}
		K&=\int_0^1\mathrm{d}t\left((1-t^2)(1-(\alpha t/\beta)^2)\right)^{-1/2}\;,\\
		K'&=\int_0^{\infty}\mathrm{d}t\left((1+t^2)(1+(\alpha t/\beta)^2)\right)^{-1/2}\ .
	\end{eqnarray}

\noindent This change of variables correpsonds to parametrising the two-cut complex $w$-plane on the torus $\mathbb{C}/(\mathbb{Z}+\tau\mathbb{Z})$ with modular parameter 

\begin{equation}
	\label{tau}
	\tau=\I \frac{K'}{K}\;.
\end{equation}

\noindent The coordinate $w$ is invariant under $\sigma\to-\sigma$ and (anti-) periodic along the respective cycles of the torus:

\begin{equation}
	w(\sigma+m+n\tau)=(-1)^m w(\sigma)\;,\qquad(m,n)\in\mathbb{Z}^2\;.
\end{equation}

\noindent We also require the Jacobi theta functions

\begin{eqnarray}
		\vartheta_1(u|\tau)&=\frac{1}{\I}\sum_{n\in\mathbb{Z}}(-1)^n \E^{\I\pi\tau(n+1/2)^2}\E^{\I u(2n+1)}\ ,\\
		\vartheta_2(u|\tau)&=\sum_{n\in\mathbb{Z}} \E^{\I\pi\tau(n+1/2)^2}\E^{\I u(2n+1)}\ ,\\
		\vartheta_3(u|\tau)&=\sum_{n\in\mathbb{Z}}\E^{\I\pi\tau n^2}\E^{2\I u n}\ .
	\end{eqnarray}

\noindent In particular, $\vartheta_1$ is an entire function with a unique simple zero at $u=0\ \mathrm{mod}\ \mathbb{Z}\oplus\pi\tau\mathbb{Z}$, satisfying

\begin{eqnarray}
\label{theta1}
		\vartheta_1\left(u+\pi(m+n\tau)|\tau\right)&=(-1)^{mn} \E^{-\I n(\pi\tau+2u)}\vartheta_1(u|\tau)\;,\quad(m,n)\in\mathbb{Z}^2\;,\\
		\vartheta_1(-u|\tau)&=-\vartheta_1(u|\tau)\ ,\\
\vartheta_1(u/\tau|-1/\tau)&=-\sqrt{\I\tau}\E^{\I\pi u^2/\tau}\vartheta_1(u|\tau)\ .
	\end{eqnarray}

\noindent We also note the equivalent representation of $w(\sigma)$ in terms of $\vartheta_i$,

\begin{equation}
\label{theta23}
	w(\sigma)
	=\sqrt{\alpha\beta}\frac{\vartheta_2(\pi\sigma|\tau)}{\vartheta_3(\pi\sigma|\tau)}\ .
\end{equation}

\noindent Analytic continuation of $f(w(\sigma))$ requires boundary conditions on the rectangle $(0,1)\times[0,\tau)$:

\begin{enumerate}
	\item Analyticity across $[0,\alpha]\cup[\beta,\infty]$ allows us to continue $f(w(\sigma))$ to the infinite strip $(0,1)\times \I\mathbb{R}$ by
	\begin{equation}
		\label{period}
		f(w(\tau+\sigma))=f(w(\sigma))\ , \quad \sigma\in(0,1)\times[0,\tau)\ .
	\end{equation}
	\item Analyticity across $[-\beta,-\alpha]$ allows us to extend this definition to $(0,2)\times \I\mathbb{R}$ using
	\begin{equation} 
		\label{refl}
		f(w(1+\sigma))=f(w(1-\sigma))\ , \quad \sigma\in(0,1)\times \I\mathbb{R}\ .		
	\end{equation}
	
	\item Finally, using all the above, \eref{linear} implies 
	\begin{equation}
		\label{spe-ell}
		f(w(1+\sigma))=\frac{f(w(\sigma))+f(w(2+\sigma))}{2\cos(\pi\nu)}\ , \quad\sigma\in(0,2)\times \I\mathbb{R}\ .
	\end{equation}
\end{enumerate}

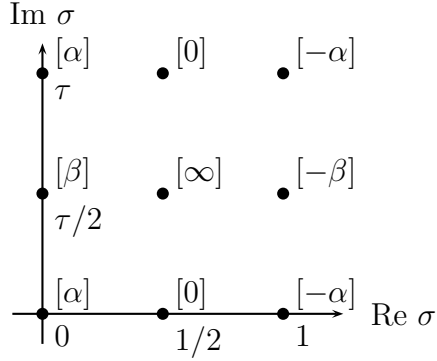
\begin{figure}[t]
	\centering
	\input{torus}
	\caption{A fundamental domain for $w\in\mathbb{C}\setminus[\pm\alpha,\pm\beta]$ is given by $\sigma\in(0,1)\times[0,\tau)$. The images $w(\sigma)$ of special points $\sigma$ are indicated in square brackets.}
	\label{fig:sheets}
\end{figure}

\noindent Solving the latter condition allows us to continue $f(w(\sigma))$ to a meromorphic function on the entire complex $\sigma$-plane, on which it satisfies two (quasi-)periodicity conditions:

\begin{eqnarray}
	\label{periodicity1}
	0&=(\E^{-\partial_{\sigma}}-\E^{\I\pi\nu})(\E^{-\partial_{\sigma}}-\E^{-\I\pi\nu})f(w(\sigma))\ ,\\
\label{periodicity2}
	0&=(\E^{-\tau\partial_{\sigma}}-1)f(w(\sigma))\ .
\end{eqnarray}

\noindent We find it convenient to follow \cite{Borot12} in introducing the unique function in $\mathrm{Ker}(\E^{-\partial_\sigma}-\E^{\I\pi\nu})$ with a simple pole of unit residue at $\sigma=0$ and no other singularities mod $\mathbb{Z}\oplus\tau\mathbb{Z}$ as

\begin{equation}
	\label{g(phi;nu)}
	g(\sigma;\nu)=\frac{\vartheta_1'(0|\tau)}{\vartheta_1( \pi\nu\tau/2|\tau)}\frac{\vartheta_1(\pi\sigma+\pi\nu\tau/2|\tau)}{\vartheta_1(\pi\sigma|\tau)}\E^{\I\pi\nu\sigma}\ ,
\end{equation}

\noindent which has a simple zero at $\sigma=-\nu\tau/2$; any solution to \eref{periodicity1} and \eref{periodicity2} may be expressed as a linear combination of derivatives $g(\sigma;\pm\nu)$ with shifted argument. The reflection relation \eref{refl} fixes the relative coefficient, so that the general solution to \eref{linear} can be expressed as

\begin{equation}
f(w)=\sum_{n\geq 0}\frac{a_n}{n!}\frac{\partial^n}{\partial\sigma_0^n}\left(\E^{-\I \pi\nu/2}g(\sigma(w)-\sigma_0;\nu)-\E^{\I\pi\nu/2}g(\sigma(w)-\sigma_0;-\nu)\right)\ ,
\end{equation}

\noindent where the requirement that $f(w)$ be free of singularities on $\mathbb{C}\setminus[\alpha,\beta]$ demands $\sigma_0=(\tau+1)/2$, and the coefficients $a_n$ are to be determined by boundary conditions supplementing the problem, using

 \begin{equation}
 \lim_{z\to\infty}(\sigma-\sigma_0)^{n+1}\frac{1}{n!}\frac{\partial^n}{\partial \sigma_0^n}g(\sigma-\sigma_0;\nu)=1\ .
 \end{equation}
 
  \noindent In particular, if $f(w)$ has a pole of order $m$ at $w=\infty$, then $a_n=0$ for $n>m$.
\\

\subsection{The limit $\alpha/\beta \searrow 0$}

 In Section \ref{sec:crit} we will be interested in the limit $\alpha/\beta\searrow 0$, in which $\tau\to \I \infty$, and thus

\begin{equation}
	\label{scalinglimit}
	\lim_{\tau\to \I \infty}w(\sigma)/\alpha=\cos(\pi\sigma)\;,\qquad \lim_{\tau\to \I \infty}g(\sigma;\nu)=\frac{\pi\ \E^{\I\pi(\nu-1)\sigma}}{\sin(\pi\nu)}\;.
\end{equation}

\noindent In this limit, $f(w)$ is holomorphic on $w/\alpha\in\mathbb{C}\setminus[1,\infty)$, and \eref{linear} becomes

 \begin{equation}
 (\E^{-\partial_{\sigma}}-\E^{\I\pi\nu})(\E^{-\partial_{\sigma}}-\E^{-\I\pi\nu})f\left(\alpha\cos(\pi \sigma)\right)=0\ .
 \end{equation}
 
 \noindent A convenient basis for the solution space is given by the Chebyshev functions. These are represented on the unit disk as

\begin{eqnarray}
		\label{chebyshevfct}
		T_{\nu}(x)&=\cos(\pi\nu\phi)\ ,\\
		U_{\nu}(x)&=\frac{\sin(\pi(\nu+1)\phi)}{\sin(\pi\phi)}\ ,\quad\ x=\cos(\pi\phi)\ .
	\end{eqnarray}

\noindent From the above definition it is easy to verify that both $T_{\nu}(x)$ and $U_{\nu}(x)$ satisfy \eref{linear}, and $T_{1/\nu}(x)$ is the functional inverse of $T_{\nu}(x)$. For non-integer $\nu$, these functions have a branch cut on $x\in [-1,-\infty)$, with discontinuity

\begin{eqnarray}
		\label{Chebyshevdiscont}
		T_{\nu}(x)_+-T_{\nu}(x)_-&=-2\I\sin(\pi\nu)\sqrt{1-x^2}\ U_{\nu-1}(-x)\ ,\\
		U_{\nu}(x)_+-U_{\nu}(x)_-&=-2\I\frac{\sin(\pi\nu)}{\sqrt{1-x^2}}T_{\nu+1}(-x)\ .
	\end{eqnarray}

\noindent When $\nu\in\mathbb{N}$, the right-hand side vanishes and we recover the definition of the Chebyshev polynomials of the first and second kind. As a result when $\nu=p/q$ is rational, $y=T_{\nu}(x)$ is the solution to the polynomial equation $T_q(y)-T_p(x)=0$. Since \eref{linear} restricts the scaling exponents $f(w)\sim(-w)^\kappa$ to the form $\kappa=2n\pm\nu$, $n \in\mathbb{Z}$ we can expand $f(w)$ as

\begin{equation}
\label{f-scaling}
	f(w)=\sum_{n\geq0}\sum_{\pm}\alpha^{2n\pm\nu}\left(t_n^{(\pm)}T_{2n\pm\nu}(-w/\alpha)+u_n^{(\pm)}U_{2n\pm\nu}(-w/\alpha)\right)\ ,
\end{equation}

\noindent with constants $t_n^{(\pm)}$, $u_n^{(\pm)}$ to be determined by boundary conditions. 

\section{Analytic Structure and Asymptotics}
\label{app:asympt}

\noindent We illustrate the analytic structure of $G_Y^{(p)}(z)$ and $G_{(p)}^Y(z)$ by graphs in which nodes depict sheets and lines between nodes depict branch cuts that connect the sheets.  Of the latter, double lines represent finite cuts and single lines represent cuts that extend to infinity.
\\

\subsection{$(q,k)=(1,2)$}
\label{ex-dimers}
From equations \eref{case1a} and \eref{case1b}, we compute the analytic structure and asymptotic behaviour of $G_{(p)}^Y(z)$,
\\

\input{sheets1a}
\\

\noindent where $\omega=\E^{\I\pi/3}$. From \eref{polynomial0}, we may compute the asymptotic behaviour of $G_{(p)}^Y(z)$ on all sheets: 
\\

\input{sheets1b}
\\

\input{sheets1c}

\subsection{$(q,k)=(2,1)$}
\label{ex-ising}
 From \eref{case2a} and \eref{case2b}, we compute the analytic structure and asymptotic behaviour of $G_{(p)}^Y(z)$,
\\

\input{sheets2a}
\\

\input{sheets2b}
\\

\noindent From the resulting polynomials $F_{(p)}(x,y)$, we may also compute the asymptotic behaviour of $G_{(p)}^Y(z)$ on all sheets. For example, from \eref{polynomial1},
\\

\input{sheets2c}
\\

\input{sheets2d}

\subsection{$(q,k)=(3,1)$}
\label{ex-potts3} 
 From \eref{case3a} and \eref{case3b}, we compute the analytic structure and asymptotic behaviour of $G_{(p)}^Y(z)$,
\\

\input{sheets3a}
\\

\input{sheets3b}
\\

\noindent From the resulting polynomials $F_{(p)}(x,y)$, we may also compute the asymptotic behaviour of $G_{(p)}^Y(z)$ on all sheets. For example, from  \eref{polynomial2},
\\

\input{sheets3inv}

\section*{References}

\bibliographystyle{iopart-num}
\bibliography{ref}
\end{document}

%% file: statesising.tex
\begin{pspicture*}(0,1)(10,4)
\psset{unit=.8cm}

\rput(3.8,4.5){$A_3$}
\psdot[dotsize=.15](3,4)
\psdot[dotsize=.15](4,4)
\psdot[dotsize=.15](5,4)
\psline(3,4)(4,4)
\psline(5,4)(4,4)

\rput(1,2.5){$A_2$}
\psdot[dotsize=.15](2,3)
\psdot[dotsize=.15](2,2)
\psline(2,3)(2,2)

\rput(3,3){$\mathbb{I}$}
\rput(4,3){$\sigma$}
\rput(5,3){$\varepsilon$}
\rput(3,2){$\varepsilon$}
\rput(4,2){$\sigma$}
\rput(5,2){$\mathbb{I}$}


\psline[linestyle=dashed](2.5,4)(5,1.5)

\rput(9,3){$
\begin{tabular}[b]{lll}
		&$A_2\times A_3$&$X_{(p|\sigma)}$\\
		\hline
		$\mathbb{I}$&$(1,1)$&$X_1$\\
		$\varepsilon$&$(2,1)$&$X_2$\\
		$\sigma$&$(2,2)$&$X_1+X_2$\\
		\hline
\end{tabular}$}

\end{pspicture*}

%% file: statespotts.tex
\begin{pspicture*}(-1,-.5)(11,4.5)
\psset{unit=.8cm}

\rput(2.8,5){$D_4$}
\psdot[dotsize=.15](2,4)
\psdot[dotsize=.15](3,4)
\psdot[dotsize=.15](4,4.5)
\psdot[dotsize=.15](4,3.5)
\psline(2,4)(3,4)
\psline(3,4)(4,4.5)
\psline(3,4)(4,3.5)

\rput(0,1.5){$A_4$}
\psdot[dotsize=.15](1,3)
\psdot[dotsize=.15](1,2)
\psdot[dotsize=.15](1,1)
\psdot[dotsize=.15](1,0)
\psline(1,3)(1,2)
\psline(1,2)(1,1)
\psline(1,1)(1,0)

\rput(2,3){$\mathbb{I}$}
\rput(3,3){$F$}
\rput(4,3){$\psi,\psi^{\dagger}$}
\rput(2,2){$\varepsilon$}
\rput(3,2){$N$}
\rput(4,2){$\sigma,\sigma^{\dagger}$}
\rput(2,1){$\varepsilon$}
\rput(3,1){$N$}
\rput(4,1){$\sigma,\sigma^{\dagger}$}
\rput(2,0){$\mathbb{I}$}
\rput(3,0){$F$}
\rput(4,0){$\psi,\psi^{\dagger}$}


\psline[linestyle=dashed](0.5,1.5)(5,1.5)

\rput(10,2.5){$
	\begin{tabular}[b]{lll}
		&$A_4\times D_4$&$X_{(p|\sigma)}$\\
		\hline
		$\mathbb{I}$&$(1,1)$&$X_1$\\
		$\psi$&$(1,3)$&$X_2$\\
		$\psi^{\dagger}$&$(1,4)$&$X_3$\\
		$F$&$(1,2)$&$X_1+X_2+X_3$\\
		$\varepsilon$&$(2,1)$&$X_2+X_3$\\
		$\sigma$&$(2,3)$&$X_1+X_3$\\
		$\sigma^{\dagger}$&$(2,4)$&$X_1+X_2$\\
		$N$&$(2,2)$&--\\
		\hline
	\end{tabular}$}
\end{pspicture*}

%% file: torus.tex
\begin{pspicture*}(1,1)(7,6)
\psset{unit=1.6cm}

\psdot[dotsize=.1](1,1)
\rput[tl](1.1,0.9){$0$}
\rput[tl](1.1,1.3){$[\alpha]$}
\psdot[dotsize=.1](1,2)
\rput[tl](1.1,1.9){$\tau/2$}
\rput[tl](1.1,2.3){$[\beta]$}
\psdot[dotsize=.1](1,3)
\rput[tl](1.1,2.9){$\tau$}
\rput[tl](1.1,3.3){$[\alpha]$}

\psdot[dotsize=.1](2,1)
\rput[tl](2.1,0.9){$1/2$}
\rput[tl](2.1,1.3){$[0]$}
\psdot[dotsize=.1](2,2)
\rput[tl](2.1,2.3){$[\infty]$}
\psdot[dotsize=.1](2,3)
\rput[tl](2.1,3.3){$[0]$}

\psdot[dotsize=.1](3,1)
\rput[tl](3.1,1.3){$[-\alpha]$}
\rput[tl](3.1,0.9){$1$}
\psdot[dotsize=.1](3,2)
\rput[tl](3.1,2.3){$[-\beta]$}
\psdot[dotsize=.1](3,3)
\rput[tl](3.1,3.3){$[-\alpha]$}

\psline{->}(0.75,1)(3.5,1)
\rput(4,1){$\mathrm{Re}\ \sigma$}
\psline{->}(1,0.75)(1,3.25)
\rput(1,3.5){$\mathrm{Im}\ \sigma$}
\end{pspicture*}

%% file: sheets1a.tex
\begin{pspicture}(8,4)
\psset{unit=.8cm}

\psline(1,1)(1,2)
\psline(1,2)(1,3)
\psline[doubleline=true](1,3)(1,4)

\psdot[dotsize=.15](1,4)
\psdot[dotsize=.15](1,3)
\psdot[dotsize=.15](1,2)
\psdot[dotsize=.15](1,1)

\rput[Br](.5,4){$G_Y^{(1)}(z)_-$}
\rput[Br](.5,3){$G_Y^{(1)}(z)_+$}
\rput[Bl](1.5,4){$z-z^{-1}+\mathcal{O}(z^{-2})$}
\rput[Bl](1.5,3){$\omega^2t_4^{-1/3}z^{1/3}-\frac{t_3}{3t_4}+\omega\frac{t_3^2-3t_2t_4}{9t_4^{5/3}}z^{-1/3}-\omega^2\frac{2t_3^3-9t_2t_3t_4}{81t_4^{7/3}}z^{-2/3}+\frac{1}{3}z^{-1}+\mathcal{O}(z^{-4/3})$}
\rput[Bl](1.5,2){$\omega t_4^{-1/3}z^{1/3}-\frac{t_3}{3t_4}+\omega^2\frac{t_3^2-3t_2t_4}{9t_4^{5/3}}z^{-1/3}-\omega\frac{2t_3^3-9t_2t_3t_4}{81t_4^{7/3}}z^{-2/3}+\frac{1}{3}z^{-1}+\mathcal{O}(z^{-4/3})$}
\rput[Bl](1.5,1){$t_4^{-1/3}z^{1/3}-\frac{t_3}{3t_4}+\frac{t_3^2-3t_2t_4}{9t_4^{5/3}}z^{-1/3}-\frac{2t_3^3-9t_2t_3t_4}{81t_4^{7/3}}z^{-2/3}+\frac{1}{3}z^{-1}+\mathcal{O}(z^{-4/3})$}
\end{pspicture}

%% file: sheets1b.tex
\begin{pspicture}(8,4)
\psset{unit=.8cm}

\rput[Br](0.5,3){$G^Y_{(0)}(z)_-$}
\rput[Br](0.5,4){$G^Y_{(0)}(z)_+$}
\rput[Bl](1.5,1){$z+t_4^{-1/3}z^{1/3}-\frac{t_3}{3t_4}+\frac{t_3^2-3(t_2-1)t_4}{9t_4^{5/3}}z^{-1/3}-\frac{2t_3^3-9(t_2-1)t_3t_4}{81t_4^{7/3}}z^{-2/3}+\frac{1}{3}z^{-1}+\mathcal{O}(z^{-4/3})$}
\rput[Bl](1.5,2){$z+\omega t_4^{-1/3}z^{1/3}-\frac{t_3}{3t_4}+\omega^2\frac{t_3^2-3(t_2-1)t_4}{9t_4^{5/3}}z^{-1/3}-\omega \frac{2t_3^3-9(t_2-1)t_3t_4}{81t_4^{7/3}}z^{-2/3}+\frac{1}{3}z^{-1}+\mathcal{O}(z^{-4/3})$}
\rput[Bl](1.5,3){$z+\omega^2 t_4^{-1/3}z^{1/3}-\frac{t_3}{3t_4}+\omega\frac{t_3^2-3(t_2-1)t_4}{9t_4^{5/3}}z^{-1/3}-\omega^2\frac{2t_3^3-9(t_2-1)t_3t_4}{81t_4^{7/3}}z^{-2/3}+\frac{1}{3}z^{-1}+\mathcal{O}(z^{-4/3})$}
\rput[Bl](1.5,4){$z^{-1}+\mathcal{O}(z^{-2})$}

\psline(1,1)(1,2)
\psline(1,2)(1,3)
\psline[doubleline=true](1,3)(1,4)
\psdot[dotsize=.15](1,1)
\psdot[dotsize=.15](1,2)
\psdot[dotsize=.15](1,3)
\psdot[dotsize=.15](1,4)

\end{pspicture}

%% file: sheets1c.tex
\begin{pspicture}(8,2)
\psset{unit=.8cm}

\psline[doubleline=true](1,1)(1,2)

\psdot[dotsize=.15](1,1)
\psdot[dotsize=.15](1,2)

\rput[Br](0.5,2){$G^Y_{(1)}(z)_+$}
\rput[Br](0.5,1){$G^Y_{(1)}(z)_-$}
\rput[Bl](1.5,1){$z+z^{-1}+\mathcal{O}(z^{-2})$}
\rput[Bl](1.5,2){$t_4z^3+t_3z^2+t_2z-z^{-1}+\mathcal{O}(z^{-2})$}

\end{pspicture}

%% file: sheets2a.tex
\begin{pspicture}(8,4)
\psset{unit=.8cm}

\psline(1,1)(1,2)
\psline[doubleline=true](1,2)(1,3)
\psline(1,3)(1,4)

\psdot[dotsize=.15](1,1)
\psdot[dotsize=.15](1,2)
\psdot[dotsize=.15](1,3)
\psdot[dotsize=.15](1,4)

\rput[Br](0.5,3){$G_Y^{(1)}(z)_-$}
\rput[Br](0.5,2){$G_Y^{(1)}(z)_+$}
\rput[Bl](1.5,4){$z+t_3^{-1/2}z^{1/2}+\frac{t_2}{2t_3}+\frac{t_2^2}{8t_3^{3/2}}z^{-1/2}-\frac{1}{2}z^{-1}+\mathcal{O}(z^{-3/2})$}
\rput[Bl](1.5,3){$z-t_3^{-1/2}z^{1/2}+\frac{t_2}{2t_3}-\frac{t_2^2}{8t_3^{3/2}}z^{-1/2}-\frac{1}{2}z^{-1}+\mathcal{O}(z^{-3/2})$}
\rput[Bl](1.5,2){$t_3^{-1/2}z^{1/2}-\frac{t_2}{2t_3}+\frac{t_2^2}{8t_3^{3/2}}z^{-1/2}+\frac{1}{2}z^{-1}+\mathcal{O}(z^{-3/2})$}
\rput[Bl](1.5,1){$-t_3^{-1/2}z^{1/2}-\frac{t_2}{2t_3}-\frac{t_2^2}{8t_3^{3/2}}z^{-1/2}+\frac{1}{2}z^{-1}+\mathcal{O}(z^{-3/2})$}
\end{pspicture}

%% file: sheets2b.tex
\begin{pspicture}(8,4)
\psset{unit=.8cm}

\psline[doubleline=true](1,1)(1,2)
\psline(1,2)(1,3)
\psline[doubleline=true](1,3)(1,4)

\psdot[dotsize=.15](1,1)
\psdot[dotsize=.15](1,2)
\psdot[dotsize=.15](1,3)
\psdot[dotsize=.15](1,4)

\rput[Br](0.5,4){$G_Y^{(2)}(z)_-$}
\rput[Br](0.5,3){$G_Y^{(2)}(z)_+$}
\rput[Bl](1.5,4){$z-z^{-1}+\mathcal{O}(z^{-2})$}
\rput[Bl](1.5,3){$2t_3^{-1/2}z^{1/2}-\frac{t_2}{t_3}+\frac{t_2^2}{4t_3^{3/2}}z^{-1/2}+\mathcal{O}(z^{-3/2})$}
\rput[Bl](1.5,2){$-2t_3^{-1/2}z^{1/2}-\frac{t_2}{t_3}-\frac{t_2^2}{4t_3^{3/2}}z^{-1/2}+\mathcal{O}(z^{-3/2})$}
\rput[Bl](1.5,1){$-z-\frac{2t_2}{t_3}+z^{-1}+\mathcal{O}(z^{-2})$}

\end{pspicture}

%% file: sheets2c.tex
\begin{pspicture}(8,3)
\psset{unit=.8cm}

\psline(1,1)(1,2)
\psline[doubleline=true](1,2)(1,3)
\psdot[dotsize=.15](1,1)
\psdot[dotsize=.15](1,2)
\psdot[dotsize=.15](1,3)

\rput[Br](0.5,3){$G^Y_{(1)}(z)_-$}
\rput[Br](0.5,2){$G^Y_{(1)}(z)_+$}
\rput[Bl](1.5,3){$t_3z^2+t_2z-z^{-1}+\mathcal{O}(z^{-2})$}
\rput[Bl](1.5,2){$z+t_3^{-1/2}z^{1/2}+\frac{t_2-1}{t_3}+\frac{(t_2-1)^2}{8t_3^{3/2}}z^{-1/2}+\frac{1}{2}z^{-1}+\mathcal{O}(z^{-3/2})$}
\rput[Bl](1.5,1){$z-t_3^{-1/2}z^{1/2}+\frac{t_2-1}{t_3}-\frac{(t_2-1)^2}{8t_3^{3/2}}z^{-1/2}+\frac{1}{2}z^{-1}-\mathcal{O}(z^{-3/2})$}

\end{pspicture}

%% file: sheets2d.tex
\begin{pspicture}(8,3)
\psset{unit=.8cm}

\psline[doubleline=true](1,1)(1,2)
\psline[doubleline=true](1,2)(1,3)
\psdot[dotsize=.15](1,1)
\psdot[dotsize=.15](1,2)
\psdot[dotsize=.15](1,3)

\rput[Br](0.5,2){$G^Y_{(2)}(z)_-$}
\rput[Br](0.5,3){$G^Y_{(2)}(z)_+$}
\rput[Bl](1.5,1){$-z-\frac{2t_2}{t_3}-z^{-1}+\mathcal{O}(z^{-2})$}
\rput[Bl](1.5,2){$\frac{t_3}{4}z^2+\frac{t_2}{2}z+\mathcal{O}(z^{-2})$}
\rput[Bl](1.5,3){$z+z^{-1}+\mathcal{O}(z^{-2})$}

\end{pspicture}

%% file: sheets3a.tex
\begin{pspicture}(8,6)
\psset{unit=.8cm}

\psline(1,1)(1,2)
\psline[doubleline=true](1,2)(1,3)
\psline(1,3)(1,4)
\psline[doubleline=true](1,4)(1,5)
\psline(1,5)(1,6)

\psdot[dotsize=.15](1,1)
\psdot[dotsize=.15](1,2)
\psdot[dotsize=.15](1,3)
\psdot[dotsize=.15](1,4)
\psdot[dotsize=.15](1,5)
\psdot[dotsize=.15](1,6)

\rput[Br](0.5,3){$G_Y^{(1)}(z)_-$}
\rput[Br](0.5,2){$G_Y^{(1)}(z)_+$}
\rput[Bl](1.5,1){$-t_3^{-1/2}z^{1/2}-\frac{t_2}{2t_3}-\frac{t_2^2}{8t_3^{3/2}}z^{-1/2}+\frac{1}{2}z^{-1}-\mathcal{O}(z^{-3/2})$}
\rput[Bl](1.5,2){$t_3^{-1/2}z^{1/2}-\frac{t_2}{2t_3}+\frac{t_2^2}{8t_3^{3/2}}z^{-1/2}+\frac{1}{2}z^{-1}+\mathcal{O}(z^{-3/2})$}
\rput[Bl](1.5,3){$ z-2t_3^{-1/2}z^{1/2}+\frac{t_2}{t_3}-\frac{t_2^2}{4t_3^{3/2}}z^{-1/2}-\mathcal{O}(z^{-3/2})$}
\rput[Bl](1.5,4){$ z+2t_3^{-1/2}z^{1/2}+\frac{t_2}{t_3}+\frac{t_2^2}{4t_3^{3/2}}z^{-1/2}+\mathcal{O}(z^{-3/2})$}
\rput[Bl](1.5,5){$2z-t_3^{-1/2}z^{1/2}+\frac{5t_2}{2t_3}-\frac{t_2^2}{8t_3^{3/2}}z^{-1/2}-\frac{1}{2}z^{-1}-\mathcal{O}(z^{-3/2})$}
\rput[Bl](1.5,6){$2z+t_3^{-1/2}z^{1/2}+\frac{5t_2}{2t_3}+\frac{t_2^2}{8t_3^{3/2}}z^{-1/2}-\frac{1}{2}z^{-1}+\mathcal{O}(z^{-3/2})$}

\end{pspicture}

%% file: sheets3b.tex
\begin{pspicture}(8,6)
\psset{unit=.8cm}

\psline[doubleline=true](1,1)(1,2)
\psline(1,2)(1,3)
\psline[doubleline=true](1,3)(1,4)
\psline(1,4)(1,5)
\psline[doubleline=true](1,5)(1,6)

\psdot[dotsize=.15](1,1)
\psdot[dotsize=.15](1,2)
\psdot[dotsize=.15](1,3)
\psdot[dotsize=.15](1,4)
\psdot[dotsize=.15](1,5)
\psdot[dotsize=.15](1,6)

\rput[Br](0.5,6){$G_Y^{(3)}(z)_-$}
\rput[Br](0.5,5){$G_Y^{(3)}(z)_+$}
\rput[Bl](1.5,6){$z-z^{-1}+\mathcal{O}(z^{-2})$}
\rput[Bl](1.5,5){$3t_3^{-1/2}z^{1/2}-\frac{3t_2}{2t_3}+\frac{3t_2^2}{8t_3^{3/2}}z^{-1/2}-\frac{1}{2}z^{-1}+\mathcal{O}(z^{-3/2})$}
\rput[Bl](1.5,4){$-3t_3^{-1/2}z^{1/2}-\frac{3t_2}{2t_3}-\frac{3t_2^2}{8t_3^{3/2}}z^{-1/2}-\frac{1}{2}z^{-1}-\mathcal{O}(z^{-3/2})$}
\rput[Bl](1.5,3){$-2z+3 t_3^{-1/2}z^{1/2}-\frac{9t_2}{2t_3}+\frac{3t_2^2}{8t_3^{3/2}}z^{-1/2}+\frac{1}{2}z^{-1}+\mathcal{O}(z^{-3/2})$}
\rput[Bl](1.5,2){$-2z-3 t_3^{-1/2}z^{1/2}-\frac{9t_2}{2t_3}-\frac{3t_2^2}{8t_3^{3/2}}z^{-1/2}+\frac{1}{2}z^{-1}+\mathcal{O}(z^{-3/2})$}
\rput[Bl](1.5,1){$-3z-6\frac{t_2}{t_3}+z^{-1}+\mathcal{O}(z^{-2})$}

\end{pspicture}

%% file: sheets3inv.tex
\begin{pspicture}(8,5)
\psset{unit=.8cm}

\psline(1,1)(1,2)
\psline[doubleline=true](1,2)(1,3)
\psline[doubleline=true](1,3)(1,4)
\psline(1,4)(1,5)

\psdot[dotsize=.15](1,1)
\psdot[dotsize=.15](1,2)
\psdot[dotsize=.15](1,3)
\psdot[dotsize=.15](1,4)
\psdot[dotsize=.15](1,5)

\rput[Br](0.5,3){$G^Y_{(2)}(z)_-$}
\rput[Br](0.5,4){$G^Y_{(2)}(z)_+$}
\rput[Bl](1.5,5){$z-t_3^{-1/2}z^{1/2}-\frac{t_2-1}{2t_3}-\frac{(t_2-1)^2}{8t_3^{3/2}}z^{1/2}+\frac{1}{2}z^{-1}-\mathcal{O}(z^{-3/2})$}
\rput[Bl](1.5,4){$z+t_3^{-1/2}z^{1/2}-\frac{t_2-1}{2t_3}+\frac{(t_2-1)^2}{8t_3^{3/2}}z^{1/2}+\frac{1}{2}z^{-1}+\mathcal{O}(z^{-3/2})$}
\rput[Bl](1.5,3){$\frac{t_3}{4}z^2+\frac{t_2}{2}z+\mathcal{O}(z^{-2})$}
\rput[Bl](1.5,2){$-z+\I t_3^{-1/2}z^{1/2}-\frac{t_2-1}{2t_3}+\I\frac{(t_2-1)^2}{8t_3^{3/2}}z^{1/2}-\frac{1}{2}z^{-1}+\mathcal{O}(z^{-3/2})$}
\rput[Bl](1.5,1){$-z-\I t_3^{-1/2}z^{1/2}-\frac{t_2-1}{2t_3}-\I\frac{(t_2-1)^2}{8t_3^{3/2}}z^{1/2}-\frac{1}{2}z^{-1}\mathcal{O}(z^{-3/2})$}

\end{pspicture}

%% file: main.bbl
\providecommand{\newblock}{}
\begin{thebibliography}{10}
\expandafter\ifx\csname url\endcsname\relax
  \def\url#1{{\tt #1}}\fi
\expandafter\ifx\csname urlprefix\endcsname\relax\def\urlprefix{URL }\fi
\providecommand{\eprint}[2][]{\href{http://arxiv.org/abs/#2}{\ttfamily #2}}

\bibitem{Tutte63}
Tutte W 1963 {\em Canad. J. Math.\/} {\bf 15} 249

\bibitem{David85}
David F 1985 {\em Nucl. Phys.\/} {\bf B257} 45

\bibitem{Ambjorn85a}
Ambjorn J, Durhuus B and Frohlich J 1985 {\em Nucl. Phys.\/} {\bf B257} 433

\bibitem{Ambjorn85b}
Ambjorn J, Durhuus B, Frohlich J and Orland P 1986 {\em Nucl. Phys.\/} {\bf
  B270} 457

\bibitem{Kazakov85}
Kazakov V~A, Migdal A~A and Kostov I~K 1985 {\em Phys. Lett.\/} {\bf B157} 295

\bibitem{DiFrancesco93}
Di~Francesco P, Ginsparg P~H and Zinn-Justin J 1995 {\em Phys. Rept.\/} {\bf
  254} 1 (\textit{Preprint} \eprint{hep-th/9306153})

\bibitem{Ginsparg93}
Ginsparg P~H and Moore G~W 1993 {Lectures on 2-D gravity and 2-D string theory}
  {\em Boulder 1992, Proceedings, Recent directions in particle theory\/}
  (\textit{Preprint} \eprint{hep-th/9304011})

\bibitem{Carroll96}
Carroll S~M, Ortiz M~E and Taylor W 1996 {\em Phys. Rev. Lett.\/} {\bf 77} 3947
  (\textit{Preprint} \eprint{hep-th/9605169})

\bibitem{Carroll97}
Carroll S~M, Ortiz M~E and Taylor W 1998 {\em Phys. Rev.\/} {\bf D58} 046006
  (\textit{Preprint} \eprint{hep-th/9711008})

\bibitem{IR10}
Ishiki G and Rim C 2010 {\em Phys. Lett.\/} {\bf B694} 272 (\textit{Preprint}
  \eprint{1006.3906})

\bibitem{BIR10}
Bourgine J~E, Ishiki G and Rim C 2010 {\em JHEP\/} {\bf 1012} 046
  (\textit{Preprint} \eprint{1010.1363})

\bibitem{Chan10}
Chan C~T, Irie H and Yeh C~H 2012 {\em Nucl. Phys.\/} {\bf B854} 67
  (\textit{Preprint} \eprint{1011.5745})

\bibitem{Atkin11}
Atkin M~R and Wheater J~F 2011 {\em JHEP\/} {\bf 1102} 084 (\textit{Preprint}
  \eprint{1011.5989})

\bibitem{Atkin12}
Atkin M~R and Zohren S 2012 {\em JHEP\/} {\bf 1211} 163 (\textit{Preprint}
  \eprint{1204.4482})

\bibitem{Potts52}
Potts R~B 1952 {\em Math. Proc. Cambridge Philos. Soc.\/} {\bf 48}

\bibitem{Behrend00}
Behrend R~E and Pearce P~A 2001 {\em J. Statist. Phys.\/} {\bf 102} 577
  (\textit{Preprint} \eprint{hep-th/0006094})

\bibitem{Kazakov88}
Kazakov V 1988 {\em Nucl. Phys. Proc. Suppl.\/} {\bf B4} 93

\bibitem{Daul94}
Daul J~M 1994  (\textit{Preprint} \eprint{hep-th/9502014})

\bibitem{ZJ99}
Zinn-Justin P 2001 {\em J. Statist. Phys.\/} {\bf 98} 245 (\textit{Preprint}
  \eprint{cond-mat/9903385})

\bibitem{Bonnet99}
Bonnet G 1999 {\em Phys. Lett.\/} {\bf B459} 575 (\textit{Preprint}
  \eprint{hep-th/9904058})

\bibitem{Eynard99}
Eynard B and Bonnet G 1999 {\em Phys. Lett.\/} {\bf B463} 273
  (\textit{Preprint} \eprint{hep-th/9906130})

\bibitem{Guionnet10}
Guionnet A, Jones V, Shlyakhtenko D and Zinn-Justin P 2012 {\em Commun. Math.
  Phys.\/} {\bf 316} 45 (\textit{Preprint} \eprint{1012.0619})

\bibitem{Bernardi11}
Bernardi O and Bousquet-M\'elou M 2011 {\em J. Combin. Theory Ser.\/} {\bf
  B101} 315

\bibitem{Borot12}
{Borot} G, {Bouttier} J and {Guitter} E 2012 {\em J. Phys. A: Math. Theor.\/}
  {\bf 45} 494017 (\textit{Preprint} \eprint{1207.4878})

\bibitem{Speicher93}
Speicher R 1993 {\em Publ. RIMS\/} {\bf 29} 731

\bibitem{Speicher94}
Speicher R 1994 {\em Math. Ann.\/} {\bf 300} 97

\bibitem{Voiculescu92}
Voiculescu D, Dykema K and Nica A {\em Free Random Variables\/} CRM Monograph
  Series (American Mathematical Soc.) ISBN 9780821869703

\bibitem{ZJ99-2}
{Zinn-Justin} P 1999 {\em Phys. Rev. E\/} {\bf 59} 4884 (\textit{Preprint}
  \eprint{math-ph/9810010})

\bibitem{ZJ98}
Zinn-Justin P 1998 {\em Commun. Math. Phys.\/} {\bf 194} 631

\bibitem{Eynard03}
Eynard B 2003 {\em JHEP\/} {\bf 0311} 018 (\textit{Preprint}
  \eprint{hep-th/0309036})

\bibitem{Pasquier86}
Pasquier V 1987 {\em Nucl. Phys.\/} {\bf B285} 162

\bibitem{Kostov95}
Kostov I~K 1996 {\em Nucl. Phys. Proc. Suppl.\/} {\bf 45A} 13
  (\textit{Preprint} \eprint{hep-th/9509124})

\bibitem{Matytsin93}
Matytsin A 1994 {\em Nucl. Phys.\/} {\bf B411} 805 (\textit{Preprint}
  \eprint{hep-th/9306077})

\bibitem{Harishchandra56}
Harish-Chandra 1956 {\em Nat. Acad. Sci.\/} {\bf 42} 538

\bibitem{Itzykson80}
Itzykson C and Zuber J 1980 {\em J. Math. Phys.\/} {\bf 21} 411

\bibitem{Gradshteyn07}
Gradshteyn I~S and Ryzhik I~M 2007 {\em Table of integrals, series, and
  products\/} seventh ed (Elsevier/Academic Press, Amsterdam)

\bibitem{Gross91}
Gross D~J and Newman M~J 1991 {\em Phys. Lett.\/} {\bf B266} 291

\bibitem{Eynard95}
Eynard B and Kristjansen C 1996 {\em Nucl. Phys.\/} {\bf B466} 463
  (\textit{Preprint} \eprint{hep-th/9512052})

\bibitem{Zee96}
Zee A 1996 {\em Nucl. Phys.\/} {\bf B474} 726 (\textit{Preprint}
  \eprint{cond-mat/9602146})

\bibitem{Staudacher89}
Staudacher M 1990 {\em Nucl. Phys.\/} {\bf B336} 349

\bibitem{Boulatov86}
Boulatov D and Kazakov V 1987 {\em Phys. Lett.\/} {\bf B186} 379

\bibitem{Kazakov86}
Kazakov V 1986 {\em Phys. Lett.\/} {\bf A119} 140

\bibitem{Cardy89}
Cardy J~L 1989 {\em Nucl. Phys.\/} {\bf B324} 581

\bibitem{Eynard92}
Eynard B and Zinn-Justin J 1992 {\em Nucl. Phys.\/} {\bf B386} 558
  (\textit{Preprint} \eprint{hep-th/9204082})

\bibitem{Eynard02}
Eynard B 2003 {\em JHEP\/} {\bf 0301} 051 (\textit{Preprint}
  \eprint{hep-th/0210047})

\bibitem{ksm13}
Atkin M, Niedner B and Wheater J 2015 {\em Springer Proc. Phys.\/} {\bf 170}
  387 (\textit{Preprint} \eprint{1511.01525})

\bibitem{Baxter71}
Baxter R 1971 {\em Stud. Appl. Math.\/} {\bf 50} 51

\bibitem{Baxter73}
Baxter R~J 1973 {\em J. Phys. C\/} {\bf 6} L445

\bibitem{Durhuus96}
Durhuus B and Kristjansen C 1997 {\em Nucl. Phys.\/} {\bf B483} 535
  (\textit{Preprint} \eprint{hep-th/9609008})

\bibitem{Cappelli86}
Cappelli A, Itzykson C and Zuber J~B 1987 {\em Nucl. Phys.\/} {\bf B280} 445

\bibitem{Caldeira03}
Caldeira A~F, Kawai S and Wheater J~F 2003 {\em JHEP\/} {\bf 0308} 041
  (\textit{Preprint} \eprint{hep-th/0306082})

\bibitem{Affleck98}
{Affleck} I, {Oshikawa} M and {Saleur} H 1998 {\em J. Phys. A: Math. Gen.\/}
  {\bf 31} 5827 (\textit{Preprint} \eprint{cond-mat/9804117})

\bibitem{Kramers41}
Kramers H and Wannier G 1941 {\em Phys. Rev.\/} {\bf 60} 252

\bibitem{BEH01a}
Bertola M, Eynard B and Harnad J~P 2002 {\em Commun. Math. Phys.\/} {\bf 229}
  73 (\textit{Preprint} \eprint{nlin/0108049})

\bibitem{BEH01b}
Bertola M, Eynard B and Harnad J 2003 {\em Theor. Math. Phys.\/} {\bf 134} 27
  (\textit{Preprint} \eprint{nlin/0112006})

\bibitem{eynard11b}
{Borot} G and {Eynard} B 2011 {\em J. Stat. Mech. Theor. Exp.\/} {\bf 1} 10
  (\textit{Preprint} \eprint{0910.5896})

\bibitem{Eynard07}
Eynard B and Orantin N 2007 {\em Commun. Number Theory Phys.\/} {\bf 1} 347

\bibitem{Zamolodchikov86}
Zamolodchikov A 1986 {\em JETP Lett.\/} {\bf 43} 730

\bibitem{Affleck91}
Affleck I and Ludwig A~W 1991 {\em Phys. Rev. Lett.\/} {\bf 67} 161

\bibitem{Friedan03}
Friedan D and Konechny A 2004 {\em Phys. Rev. Lett.\/} {\bf 93} 030402
  (\textit{Preprint} \eprint{hep-th/0312197})

\bibitem{Wilshin08}
Kawamoto S, Wheater J~F and Wilshin S 2008 {\em Int. J. Mod. Phys.\/} {\bf A23}
  2257

\end{thebibliography}
